\documentclass[12pt]{article}
\usepackage{lineno}
\usepackage{amssymb}
\usepackage{multirow}
\usepackage{epsfig,rotating}
\usepackage{amssymb,amsmath}
\usepackage{latexsym}
\usepackage{tabularx}
\usepackage{array}
\usepackage{natbib}
\usepackage[retainorgcmds]{IEEEtrantools}
\usepackage{url}            
\usepackage{booktabs}       
\usepackage{threeparttable}
\usepackage{amsfonts}       
\usepackage{nicefrac}       
\usepackage{microtype}      
\usepackage{lipsum}
\usepackage{amsmath}
\usepackage{amssymb}
\usepackage{indentfirst}
\usepackage{mathtools}
\usepackage{algorithm}
\usepackage{algpseudocode}
\usepackage{appendix}
\usepackage{hyperref} 
\hypersetup{
	colorlinks   = true, 
	urlcolor     = blue, 
	linkcolor    = blue, 
	citecolor   = blue 
}
\usepackage{nccmath}
\usepackage{centernot}
\usepackage{setspace}

\renewcommand{\baselinestretch}{1.5}

\usepackage{tkz-graph}  
\usetikzlibrary{shapes.geometric}
\usepackage{siunitx}
\usetikzlibrary{positioning}
\tikzset{mynode/.style={draw,text width=1in,align=center}
}

\usepackage{bm,fullpage,enumerate,setspace}
	
\newcommand{\indep}{\perp \!\!\! \perp}

\begin{document}

\def\reduce#1{{\small  #1}}
\def\c#1{\ensuremath{\mathcal{#1}}}
\def\foot#1{\footnote{#1}}

\newcommand{\nmathbf}{\bm}

\def\bfA{\nmathbf A}
\def\bfB{\nmathbf B}
\def\bfC{\nmathbf C}
\def\bfD{\nmathbf D}
\def\bfE{\nmathbf E}
\def\bfF{\nmathbf F}
\def\bfG{\nmathbf G}
\def\bfH{\nmathbf H}
\def\bfI{\nmathbf I}
\def\bfJ{\nmathbf J}
\def\bfK{\nmathbf K}
\def\bfL{\nmathbf L}
\def\bfM{\nmathbf M}
\def\bfN{\nmathbf N}
\def\bfO{\nmathbf O}
\def\bfP{\nmathbf P}
\def\bfQ{\nmathbf Q}
\def\bfR{\nmathbf R}
\def\bfS{\nmathbf S}
\def\bfT{\nmathbf T}
\def\bfU{\nmathbf U}
\def\bfV{\nmathbf V}
\def\bfW{\nmathbf W}
\def\bfX{\nmathbf X}
\def\bfY{\nmathbf Y}
\def\bfZ{\nmathbf Z}

\def\bfa{\nmathbf a}
\def\bfb{\nmathbf b}
\def\bfc{\nmathbf c}
\def\bfd{\nmathbf d}
\def\bfe{\nmathbf e}
\def\bff{\nmathbf f}
\def\bfg{\nmathbf g}
\def\bfh{\nmathbf h}
\def\bfi{\nmathbf i}
\def\bfj{\nmathbf j}
\def\bfk{\nmathbf k}
\def\bfl{\nmathbf l}
\def\bfm{\nmathbf m}
\def\bfn{\nmathbf n}
\def\bfo{\nmathbf o}
\def\bfp{\nmathbf p}
\def\bfq{\nmathbf q}
\def\bfr{\nmathbf r}
\def\bfs{\nmathbf s}
\def\bft{\nmathbf t}
\def\bfu{\nmathbf u}
\def\bfv{\nmathbf v}
\def\bfw{\nmathbf w}
\def\bfx{\nmathbf x}
\def\bfy{\nmathbf y}
\def\bfz{\nmathbf z}

\def\bfalpha  {\nmathbf \alpha}
\def\bfbeta   {\nmathbf \beta}
\def\bfgamma  {\nmathbf \gamma}
\def\bfdelta  {\nmathbf \delta}
\def\bfepsilon{\nmathbf \epsilon}
\def\bfzeta   {\nmathbf \zeta}
\def\bfeta    {\nmathbf \eta}
\def\bftheta  {\nmathbf \theta}
\def\bfiota   {\nmathbf \iota}
\def\bfkappa  {\nmathbf \kappa}
\def\bflambda {\nmathbf \lambda}
\def\bfmu     {\nmathbf \mu}
\def\bfnu     {\nmathbf \nu}
\def\bfxi     {\nmathbf \xi}
\def\bfomicron{\nmathbf \omicron}
\def\bfpi     {\nmathbf \pi}
\def\bfrho    {\nmathbf \rho}
\def\bfsigma  {\nmathbf \sigma}
\def\bftau    {\nmathbf \tau}
\def\bfupsilon{\nmathbf \upsilon}
\def\bfphi    {\nmathbf \phi}
\def\bfpsi    {\nmathbf \psi}
\def\bfchi    {\nmathbf \chi}
\def\bfomega  {\nmathbf \omega}

\def\bfAlpha  {\nmathbf \Alpha}
\def\bfBeta   {\nmathbf \Beta}
\def\bfGamma  {\nmathbf \Gamma}
\def\bfDelta  {\nmathbf \Delta}
\def\bfEpsilon{\nmathbf \Epsilon}
\def\bfZeta   {\nmathbf \Zeta}
\def\bfEta    {\nmathbf \Eta}
\def\bfTheta  {\nmathbf \Theta}
\def\bfIota   {\nmathbf \Iota}
\def\bfKappa  {\nmathbf \Kappa}
\def\bfLambda {\nmathbf \Lambda}
\def\bfMu     {\nmathbf \Mu}
\def\bfNu     {\nmathbf \Nu}
\def\bfXi     {\nmathbf \Xi}
\def\bfOmicron{\nmathbf \Omicron}
\def\bfPi     {\nmathbf \Pi}
\def\bfRho    {\nmathbf \Rho}
\def\bfSigma  {\nmathbf \Sigma}
\def\bfTau    {\nmathbf \Tau}
\def\bfUpsilon{\nmathbf \Upsilon}
\def\bfPhi    {\nmathbf \Phi}
\def\bfPsi    {\nmathbf \Psi}
\def\bfChi    {\nmathbf \Chi}
\def\bfOmega  {\nmathbf \Omega}

\newcommand{\ttheta}{\tilde{\theta}}
\newcommand{\bfzero}{{\nmathbf 0}}
\newcommand{\bfone}{{\nmathbf 1}}
\newcommand{\vareps}{\varepsilon}
\def\bfvareps{\nmathbf \varepsilon}
\newcommand{\tgamma}{\tilde\gamma}

\newcommand{\cfA}{\mbox{\c{A}}}
\newcommand{\cfB}{\mbox{\c{B}}}
\newcommand{\cfC}{\mbox{\c{C}}}
\newcommand{\cfD}{\mbox{\c{D}}}
\newcommand{\cfE}{\mbox{\c{E}}}
\newcommand{\cfF}{\mbox{\c{F}}}
\newcommand{\cfG}{\mbox{\c{G}}}
\newcommand{\cfH}{\mbox{\c{H}}}
\newcommand{\cfI}{\mbox{\c{I}}}
\newcommand{\cfJ}{\mbox{\c{J}}}
\newcommand{\cfK}{\mbox{\c{K}}}
\newcommand{\cfL}{\mbox{\c{L}}}
\newcommand{\cfM}{\mbox{\c{M}}}
\newcommand{\cfN}{\mbox{\c{N}}}
\newcommand{\cfO}{\mbox{\c{O}}}
\newcommand{\cfP}{\mbox{\c{P}}}
\newcommand{\cfQ}{\mbox{\c{Q}}}
\newcommand{\cfR}{\mbox{\c{R}}}
\newcommand{\cfS}{\mbox{\c{S}}}
\newcommand{\cfT}{\mbox{\c{T}}}
\newcommand{\cfU}{\mbox{\c{U}}}
\newcommand{\cfV}{\mbox{\c{V}}}
\newcommand{\cfX}{\mbox{\c{X}}}
\newcommand{\cfY}{\mbox{\c{Y}}}
\newcommand{\cfZ}{\mbox{\c{Z}}}

\def\boldfacefake#1{\kern-4pt
   \hbox{ \mathsurround=0pt
   \hbox to 0.4pt{$#1$\hss}\hbox to 0.4pt{$#1$\hss}\hbox {$#1$}}}

\def\bfitI{\mbox{\boldfacefake{\it I}}}

\newcommand{\bcfA}{\boldsymbol{\mathcal{A}}}
\newcommand{\bcfB}{\boldsymbol{\mathcal{B}}}
\newcommand{\bcfC}{\boldsymbol{\mathcal{C}}}
\newcommand{\bcfD}{\boldsymbol{\mathcal{D}}}
\newcommand{\bcfE}{\boldsymbol{\mathcal{E}}}
\newcommand{\bcfF}{\boldsymbol{\mathcal{F}}}
\newcommand{\bcfG}{\boldsymbol{\mathcal{G}}}
\newcommand{\bcfH}{\boldsymbol{\mathcal{H}}}
\newcommand{\bcfI}{\boldsymbol{\mathcal{I}}}
\newcommand{\bcfJ}{\boldsymbol{\mathcal{J}}}
\newcommand{\bcfK}{\boldsymbol{\mathcal{K}}}
\newcommand{\bcfL}{\boldsymbol{\mathcal{L}}}
\newcommand{\bcfM}{\boldsymbol{\mathcal{M}}}
\newcommand{\bcfN}{\boldsymbol{\mathcal{N}}}
\newcommand{\bcfO}{\boldsymbol{\mathcal{O}}}
\newcommand{\bcfP}{\boldsymbol{\mathcal{P}}}
\newcommand{\bcfQ}{\boldsymbol{\mathcal{Q}}}
\newcommand{\bcfR}{\boldsymbol{\mathcal{R}}}
\newcommand{\bcfS}{\boldsymbol{\mathcal{S}}}
\newcommand{\bcfT}{\boldsymbol{\mathcal{T}}}
\newcommand{\bcfU}{\boldsymbol{\mathcal{U}}}
\newcommand{\bcfV}{\boldsymbol{\mathcal{V}}}
\newcommand{\bcfW}{\boldsymbol{\mathcal{W}}}
\newcommand{\bcfX}{\boldsymbol{\mathcal{X}}}
\newcommand{\bcfY}{\boldsymbol{\mathcal{Y}}}
\newcommand{\bcfZ}{\boldsymbol{\mathcal{Z}}}


\newcommand{\p}{\mbox{P}}
\newcommand{\D}{\mbox{D}}
\newcommand{\E}{\mbox{E}}
\newcommand{\Mo}{\mbox{Mo}}
\newcommand{\Me}{\mbox{Me}}
\newcommand{\Cov}{\mbox{Cov}}
\newcommand{\Var}{\mbox{Var}}
\newcommand{\Corr}{\mbox{Corr}}
\newcommand{\Q}{\mbox{Q}}
\newcommand{\vech}{\mbox{vech}}
\newcommand{\tr}{\mbox{tr}}

\newcommand{\Bb}{\mbox{Bb}}
\newcommand{\Be}{\mbox{Be}}
\newcommand{\Bi}{\mbox{Bi}}
\newcommand{\Br}{\mbox{Br}}
\newcommand{\Ca}{\mbox{Ca}}
\newcommand{\Di}{\mbox{Di}}
\newcommand{\Ex}{\mbox{Ex}}
\newcommand{\Fs}{\mbox{Fs}}
\newcommand{\Ga}{\mbox{Ga}}
\newcommand{\Ge}{\mbox{Ge}}
\newcommand{\Gg}{\mbox{Gg}}
\newcommand{\Hy}{\mbox{Hy}}
\newcommand{\Ig}{\mbox{Ig}}
\newcommand{\Ip}{\mbox{Ip}}
\newcommand{\Lo}{\mbox{Lo}}
\newcommand{\Mu}{\mbox{Mu}}
\newcommand{\Nb}{\mbox{Nb}}
\newcommand{\Ng}{\mbox{Ng}}
\newcommand{\Nw}{\mbox{Nw}}
\newcommand{\Po}{\mbox{Po}}
\newcommand{\Pg}{\mbox{Pg}}
\newcommand{\Pn}{\mbox{Pn}}
\newcommand{\Ra}{\mbox{Ra}}
\newcommand{\St}{\mbox{St}}
\newcommand{\Un}{\mbox{Un}}
\newcommand{\Wi}{\mbox{Wi}}

\newcommand{\dd}[1]{\,d#1}
\newcommand{\barx}{\mbox{$\overline x$}}
\newcommand{\comb}[2]{{#1\choose#2}}
\newcommand{\ontop}[2]{{#1\atop#2}}
\newcommand{\h}{\hbox{$1\over2$}}
\newcommand{\ok}{\hfill\fbox{}}

\newcommand{\brow}[2]{\mbox{$\{{#1}_1,\ldots,{#1}_{#2}\}$}}
\newcommand{\prow}[2]{\mbox{$({#1}_1,\ldots,{#1}_{#2})$}}
\newcommand{\row}[2]{\mbox{${#1}_1,\ldots,{#1}_{#2}$}}
\newcommand{\data}{\row{x}{n}}
\newcommand{\bdata}{\prow{x}{n}}
\newcommand{\ie}{\emph{i.e.},\ }
\newcommand{\co}{\emph{cf.}\ }
\newcommand{\eg}{\emph{e.g.}, }
\newcommand{\etalc}{\emph{et al.},\ }
\newcommand{\etal}{\emph{et al.}\ }

\newenvironment{mat}{\left(\begin{array}}{\end{array}\right)}

\newcommand{\twomat}[5]{\begin{array}{ll}\displaystyle
              #1&\mbox{#2}\\[#5pt]#3&\mbox{#4}\end{array}}

\newcommand{\twocases}[6]
             {#1=\left\{\twomat{#2}{#3}{#4}{#5}{#6}\right.}

\newcommand{\mymatrix}[4]
           {\left(\begin{array}{ll}{#1} & {#2}\\
                                   {#3} & {#4}
                    \end{array} \right)}

\newcommand{\btable}{\begin{table}[h]\centering}
\newcommand{\etable}{\end{table}}
\newcommand{\bt}{\begin{parag}\small \let\b=\nsb \let\sb=\nssb \begin{tabular}}
\newcommand{\et}{\end{tabular}\let\b=\nb \let\sb=\nsb\end{parag}}
\newcommand{\capt}[1]
         {\begin{quotation}\caption{{\small #1 }}\vs{-2}\end{quotation}}

\newenvironment{parag}{\par}{\par}
\newenvironment{dif}
    {\begin{parag}\small \let\b=\nsb \let\sb=\nssb \begin{parag}}
    {\let\b=\nb \let\sb=\nsb \end{parag}\end{parag}}


\newcommand{\be}{\begin{eqnarray}}
\newcommand{\ee}{\end{eqnarray}}
\newcommand{\ba}{\begin{eqnarray*}}
\newcommand{\ea}{\end{eqnarray*}}

\newcommand{\go}{\rightarrow}
\newcommand{\goi}{\rightarrow \infty}
\newcommand{\ol}{\overline}
\newcommand{\fr}{\frac}
\newcommand{\pn}{\par\noindent}
\newcommand{\nc}{\nonumber\\}
\newcommand{\ssum}{\mbox{$\sum$}}
\newcommand{\hhline}{\hline\hline}

\newtheorem{theorem0}{Theorem}
\newtheorem{lemma0}{Lemma}
\newtheorem{remark0}{Remark}
\newtheorem{fact0}{Fact}
\newtheorem{example0}{Example}
\newtheorem{definition0}{Definition}
\newtheorem{corollary0}{Corollary}
\newtheorem{proposition0}{Proposition}
\newtheorem{algorithmY}{Algorithm}

\newenvironment{theorem}{\begin{theorem0} \mbox{} }{\end{theorem0}}
\newenvironment{lemma}{\begin{lemma0} \mbox{}}{\end{lemma0}}
\newenvironment{remark}{\begin{remark0} \mbox{}}{\end{remark0}}
\newenvironment{fact}{\begin{fact0} \mbox{}}{\end{fact0}}
\newenvironment{example}{\begin{example0} }{\end{example0}}
\newenvironment{definition}{\begin{definition0} \mbox{}}{\end{definition0}}
\newenvironment{corollary}{\begin{corollary0} \mbox{} }{\end{corollary0}}
\newenvironment{proposition}{\begin{proposition0}\mbox{} }{\end{proposition0}}
\newenvironment{algorithm1}{\begin{algorithmY}\mbox{} }{\end{algorithmY}}

\newcommand{\reals}{\mbox{\rm I\kern-.20em R}}
\newcommand{\sreals}{\mbox{\small \rm I\kern-.20em R}}
\newcommand{\mylinel}{\renewcommand{\baselinestretch}{1.8}\tiny\small}
\newcommand{\goto}{\rightarrow}
\newcommand{\expect}{\E}

\newcommand{\bdfn}{\begin{dfn}}
\newcommand{\edfn}{\end{dfn}}
\newcommand{\bteo}{\begin{teo}}
\newcommand{\eteo}{\end{teo}}
\newcommand{\bexa}{\begin{exa}}
\newcommand{\eexa}{\end{exa}}
\newcommand{\bdif}{\begin{dif}}
\newcommand{\edif}{\end{dif}}
\newcommand{\bpro}{\begin{proof}}
\newcommand{\epro}{\end{proof}}


\thispagestyle{empty}

\begin{spacing}{1}
\begin{center}
{\bf\Large Cluster-induced target shift and synthetic approximation for high-dimensional clustered data 
\\\vspace*{0.1in} \vspace*{0.1in}}
\end{center}

\begin{center}
Shangyuan Ye$^{1,\dag}$, Cong Zhang$^{2}$, Ying Chen$^{3}$, Guanbo Wang$^{4}$, Ye Liang$^{5}$
\vspace*{0.15in}

$^{1}$Department of Mathematics and Statistics, Florida International University \\
$^2$School of Economics, University of Nottingham Ningbo China \\
$^3$Department of Statistics and Data Science, University of Pennsylvania \\
$^4$The Dartmouth Institute for Health Policy and Clinical Practice, Geisel School of Medicine, Dartmouth College \\
$^5$Department of Statistics, Oklahoma State University \\
\vspace*{0.1in}

\end{center}    
\end{spacing}




\begin{quotation}
\small

\begin{spacing}{1}
\noindent {\bf Abstract}
In high-dimensional clustered data, covariates whose distributions differ across clusters can act as proxies for unobserved cluster effects. We show that the marginal-model LASSO implicitly uses sparse combinations of such covariates to shift the estimation target from structural coefficients to a contaminated vector, which inflates false selections. We therefore propose the Synthetic Heterogeneous-Effects LASSO (SHEL), which augments the regression design with separately penalized, outcome-independent, cluster-constant synthetic covariates to correct the cluster-induced target shift. Under a fixed nuisance-effects formulation, we derive oracle properties and establish consistency of the structural coefficients when the residual heterogeneity is weakly aligned with the augmented design. We also demonstrate the distinction between the structural parameter under small residual heterogeneity and the population working parameter when residual heterogeneity persists. Further, we develop a debiased estimator with a cluster-robust sandwich variance, and establish asymptotic results. Simulations show far fewer false selections than the marginal LASSO, and an analysis of longitudinal neutrophil transcriptomic data from hospitalized COVID-19 patients yields a more parsimonious gene set that retains the severity markers of the original study. 

\vspace*{0.5in}

\end{spacing}

\end{quotation}

\vspace*{0.3in}

$^{\dag}$ Corresponding author. E-mail: \textit{sye@fiu.edu}

\newpage
\setcounter{page}{1}

\section{Introduction}
Clustered data, in which observations are grouped into units such as patients, hospitals, or regions, are common in sociology, biomedicine, and health services research. Although methods for the resulting multilevel structure are well developed when the number of covariates is small, high dimensionality introduces a difficulty that is specific to clustered data and has received little attention. Specifically, when the distributions of many covariates differ across clusters, a penalized estimator can recruit sparse combinations of those covariates as proxies for the cluster effects. As a result, the estimation target shifts away from the structural regression coefficient, and the selected variables are contaminated by the proxy variable. This paper provides theoretical reasoning and a solution to this phenomenon. 

Existing high-dimensional methods for clustered data include generalized linear mixed-effects models (GLMM) and marginal models fit through generalized estimation equations (GEE), both of which incorporate regularization such as the LASSO \citep{tibshirani1996regression} or SCAD \citep{fan2001variable} for dimension reduction. High-dimensional GLMMs \citep{fan2012variable,li2015variable,hui2017hierarchical} model cluster-specific heterogeneity through latent random effects and estimate the fixed effects using penalized likelihood or related procedures. Conventional formulations typically assume that the random effects are independent of the covariates; when this exogeneity assumption is violated but not modeled, the resulting fixed-effect estimators may be distorted. Moreover, likelihood-based estimation requires integrating, or approximating integration, over the latent random effects, which can become computationally demanding as the model dimension increases \citep{tutz2013likelihood,ye2025variable}. Marginal model-based approaches \citep{xu2010gee,wang2012penalized,chen2026variable}, in contrast, specify the marginal mean directly and account for within-cluster dependence through a working covariance structure. They therefore avoid specification and integration of a latent random-effects distribution and are particularly attractive when the marginal association between the response and covariates is the inferential target. However, when cluster-level heterogeneity is associated with covariate levels, this marginal target need not coincide with the structural regression effect of interest.

Marginal-model variable selection procedures have been observed to include a large proportion of falsely selected variables \citep{ye2025model,muff2016marginal}. We formally explain this phenomenon in Section \ref{subsec:shift}. Informally, when the distributions of some covariates vary across clusters, their cluster-level components may carry information about the omitted heterogeneous effects, and a penalized estimator can use a sparse collection of such covariates to approximate the cluster effects even when the corresponding individual-level covariates have no structural effect on the outcome. We refer to this mechanism as \textit{surrogate selection}. Surrogate selection may improve prediction, but it shifts the estimation target away from the structural coefficient toward a coefficient vector that also reflects the covariates' ability to proxy the omitted cluster effects; we refer to this consequence as \textit{target shift}. We also characterize the conditions under which the marginal-model LASSO concentrates around such a shifted target.

Motivated by the target-shift phenomenon, we propose Synthetic Heterogeneous-Effects LASSO (SHEL) and its extension to generalized linear models (GSHEL). Rather than allowing the individual-level covariates themselves to absorb cluster heterogeneity, SHEL augments the regression with synthetic covariates constructed from the covariates that are cluster-constant, so that they can correct the target shift induced by the between-cluster heterogeneity, allowing the original covariates to retain their role in estimating individual-level regression effects. The proposed framework extends the within-between decomposition device \citep{mundlak1978pooling,neuhaus1998between} to high-dimensional generalized linear models with sparse, outcome-independent dictionaries whose composition is learned from the covariates alone. Importantly, the validity of the proposed approach does not require consistent recovery of a unique synthetic coefficient vector. Instead, the theory requires only that the constructed synthetic dictionary provides a sufficiently low-complexity representation of the nuisance component while preserving the relevant restricted eigenvalue condition. 

Our theoretical results are developed under a fixed nuisance-effects formulation. We first characterize the target shift of the marginal-model LASSO and show how it arises from the trade-off between approximating between-cluster heterogeneity and introducing within-cluster leakage. For the proposed framework, we derive oracle bounds in which the residual synthetic approximation error appears as an explicit term, and show that the usual high-dimensional rates and consistency for the structural coefficient are recovered whenever the residual heterogeneity is weakly aligned with the augmented design. When residual heterogeneity persists, we clarify the estimand: in nonlinear models, the estimator targets a marginal working parameter, and we bound its distance from the structural coefficient by the residual approximation error. 

This paper further considers statistical inference for the proposed GSHEL estimator. Existing post-selection inference-based methods, such as data splitting \citep{rianldo2019bootstrapping}, simultaneous inference \citep{berk2013valid}, and selective inference \citep{lee2016exact}, were developed for independent observations and do not transfer directly to clustered data. 
We adapt the debiased inference approach to the full Fisher-weighted augmented design with a cluster-robust sandwich variance estimator \citep{van2014asymptotically} and establish asymptotic normality for the marginal working parameter. In simulations, the proposed methods identify all active covariates while reducing the false selections of the marginal-model LASSO substantially, and the debiased procedure maintains type I error near the nominal level when the cluster effects depend on the covariates, a setting in which a low-dimensional mixed-effects refit becomes markedly anti-conservative. 

The remainder of the paper is organized as follows. Section \ref{sec:lmm} characterizes surrogate selection and the resulting target shift of the marginal-model LASSO, introduces SHEL and GSHEL, and establishes the within- and between-cluster interpretation and identifiability of the proposed adjustment. Section \ref{sec:thm} establishes the theoretical properties of the proposed framework. Section \ref{sec:ps} develops debiased inference. Section \ref{sec:sim} evaluates finite-sample performance through simulation studies, and Section \ref{sec:real} applies the method to a longitudinal bulk RNA-seq dataset. Finally, Section \ref{sec:disc} concludes with a discussion.

\section{Cluster-induced target shift and synthetic adjustment} \label{sec:lmm}
\subsection{Model and notation}
Let $Y_{ij}$ denote the response variable of the $j$th individual in the $i$th cluster, and $\bfX_{ij}$ denote the corresponding $p$-dimensional vector of covariates for $i \in \{1, \ldots, m\}$ and $j \in \{1, \ldots, n_i\}$, with $N = n_1 + \ldots + n_m$.
We consider the generalized linear heterogeneous-intercepts model, where the conditional distribution of $Y_{ij}$ given $\bfX_{ij}$ belongs to an exponential family:
\be \label{glmm}
f(y_{ij}; \eta_{ij}) = \exp\{ y_{ij} \eta_{ij} - b(\eta_{ij}) + c(y_{ij}) \}, \quad \eta_{ij} = \bfX_{ij}^\top \bfbeta^0 + \alpha_i,
\ee
where $\bfbeta^0$ is a $p$-vector of fixed effects, $\bfalpha = (\alpha_1, \ldots, \alpha_m)^\top$ is an $m$-vector of cluster-specific intercepts, and $b(\cdot)$ and $c(\cdot)$ are known functions. Throughout the theoretical development, the cluster-specific intercepts $\alpha_i$ are treated as fixed nuisance effects, while $\bfbeta^0$ is the structural regression coefficient of interest.

We begin with the linear heterogeneous-intercepts model, the linear analogue of (\ref{glmm}),
\be \label{lmm}
Y_{ij} = \bfX_{ij}^\top \bfbeta^0 + \alpha_i + \epsilon_{ij},
\ee
where $\epsilon_{ij}$ are independent sub-Gaussian random errors with mean $0$ and sub-Gaussian norm bounded by a constant $\sigma > 0$, that is, $\|\epsilon_{ij}\|_{\psi_2} \le \sigma$ for all $i$ and $j$. Here, $\sigma$ is a uniform upper bound on the tail magnitude of the within-cluster random errors. In high-dimensional settings, we assume that the true model is sparse, i.e., only a small number of $\beta_{l}^0$, $l = 1,...,p$ values are non-zero. We denote $\cfM_1 = \mbox{supp} (\bfbeta^0)$, and $|\cfM_1| = M_1$. Let $\bfD \in \mathbb{R}^{N \times m}$ denote the cluster-indicator matrix, whose row for observation $(i,j)$ has a one in the $i$th column and zero elsewhere. Then the $(i,j)$-th entry of the vector $\bfD\bfalpha$ is $\alpha_i$, meaning that each observation in cluster $i$ shares the same cluster-specific intercept. In matrix form, model (\ref{lmm}) is
\be \label{lmm_m}
\bfY = \bfX \bfbeta^0 + \bfD \bfalpha + \bfepsilon.
\ee

\subsection{Surrogate selection and target shift in marginal-model LASSO} \label{subsec:shift}
The marginal version of the model (\ref{lmm}) can be expressed as
\be \label{m_lmm}
Y_{ij} = \bfX_{ij}^\top \tilde{\bfbeta} + \tilde{\epsilon}_{ij}, ~\tilde{\epsilon}_{ij} = \alpha_i + \epsilon_{ij},
\ee
in which the cluster effect is absorbed into the error. The marginal parameter $\tilde{\bfbeta}$ represents the population-average covariate effects \citep{zeger1988models}. For fixed $p$, GEE with a robust sandwich variance estimator is the standard tool for inference on $\tilde{\bfbeta}$. In high-dimensional settings, penalized GEE (PGEE) has been proposed to perform variable selection and parameter estimation simultaneously \citep{wang2012penalized}.

Motivated by the marginal representation (\ref{m_lmm}), we first study the marginal-model LASSO estimator, which ignores the cluster-specific intercepts:
\be \label{LASSO}
\hat{\bfbeta} = \underset{\bfbeta \in \mathbb{R}^{p}}{\arg\min} \left\{ \frac{1}{2N} \lVert \bfY - \bfX \bfbeta \rVert_2^2 + \lambda \lVert \bfbeta \rVert_1 \right\}.
\ee
This estimator can be viewed as the analogue of a PGEE estimator under an independence working correlation structure, specialized to the Gaussian linear model. 

Substituting (\ref{lmm_m}) into (\ref{LASSO}) shows that the objective function fits the full signal $\bfX\bfbeta^0 + \bfD\bfalpha$ rather than the structural component $\bfX\bfbeta^0$ alone, so the fitted linear predictor $\bfX\hat{\bfbeta}$ may partially absorb variations in the latent heterogeneous intercepts. To describe how the omitted term $\bfD\bfalpha$ can be represented by heterogeneously distributed covariates, define the empirical cluster mean $\bar{\bfX}_i = n_i^{-1} \sum_{j=1}^{n_i} \bfX_{ij}$ and the within-cluster deviation $\bfxi_{ij} = \bfX_{ij} - \bar{\bfX}_i$,  so that $\sum_{j=1}^{n_i} \bfxi_{ij} = \bfzero$ for every $i$. Let $\cfP_0 \subset \{1, \ldots, p\}$, with $|\cfP_0|=p_0$, index the covariates whose empirical cluster means are allowed to vary across clusters, and let $\bfB_c \in \reals^{m \times p_0}$ collect the corresponding cluster means. For any $\bfgamma \in \mathbb{R}^{p_0}$, define 
\ba 
\delta_i(\bfgamma) = \alpha_i - \bar{\bfX}_{i,\cfP_0}^{\top} \bfgamma, \quad q_{ij}(\bfgamma) = \bfxi_{ij,\cfP_0}^{\top} \bfgamma. 
\ea 
The first quantity is the between-cluster error incurred when the cluster means approximate $\alpha_i$; the second is the within-cluster leakage that arises because the proxy is built from individual-level covariates, which are not constant within clusters. For a prespecified sparsity level $M_2$, we define the sparse marginal surrogate by 
\ba 
\bfgamma^\star = \underset{\| \bfgamma \|_0 \le M_2}{\arg\min} \left\{ \Delta_N^2(\bfgamma) + Q_N^2(\bfgamma) \right\}, 
\ea 
where 
\ba 
\Delta_N^2(\bfgamma) = \sum_{i=1}^m \pi_i \delta_i^2(\bfgamma), \quad Q_N^2(\bfgamma) = \frac{1}{N} \sum_{i=1}^m\sum_{j=1}^{n_i} q_{ij}^2(\bfgamma), \quad \pi_i=\frac{n_i}{N}. 
\ea 
Since $\sum_{j=1}^{n_i} q_{ij}(\bfgamma) = 0$ for every $i$, equivalently,
\ba 
\bfgamma^\star = \underset{\| \bfgamma \|_0 \le M_2}{\arg\min} \frac{1}{N} \sum_{i=1}^m \sum_{j=1}^{n_i} \left( \alpha_i - \bfX_{ij,\cfP_0}^{\top} \bfgamma \right)^2. 
\ea
Let the shifted coefficient $\bfbeta^\star = \bfbeta^0 + \Tilde{\bfgamma}^\star$, where $\Tilde{\gamma}^\star_{l} = \gamma^\star_l$ when $l \in \cfP_0$ and 0 otherwise. Model \eqref{lmm} admits the decomposition 
\be \label{lmm_decomp} 
Y_{ij} = \bfX_{ij}^{\top} \bfbeta^\star + \delta_i^\star - q_{ij}^\star + \epsilon_{ij}, 
\ee
where $\delta_i^\star = \delta_i(\bfgamma^\star)$ and $q_{ij}^\star = q_{ij}(\bfgamma^\star)$. Representation \eqref{lmm_decomp} shows how covariate heterogeneity affects marginal-model LASSO: the latent intercept $\alpha_i$ is partially approximated, which may improve prediction, but the estimation target is shifted from the structural fixed-effects vector $\bfbeta^0$ to the contaminated coefficient vector $\bfbeta^\star$. The term $\delta_i^\star$ is the between-cluster approximation error, and $q_{ij}^\star$ is the within-cluster leakage term that results from using the individual-level covariates as a surrogate for $\alpha_i$.

Let $(\Delta_N^\star)^2 = \Delta_N^2(\bfgamma^\star)$ and $(Q_N^\star)^2 = Q_N^2(\bfgamma^\star)$, and denote $\bfr^\star = \bfD \bfdelta^\star - \bfq^\star$ for the surrogate residual, where $\bfdelta^\star$ and $\bfq^\star$ are the vectors with entries $\delta_i^\star$ and $q_{ij}^\star$. Because $\sum_{j=1}^{n_i} \bfxi_{ij} = \bfzero$, we have $(\bfD \bfdelta^\star)^\top \bfq^\star = 0$, and hence 
\ba
\frac{1}{N} \| \bfD \bfalpha - \bfX \widetilde{\bfgamma}^\star \|_2^2 = \frac{1}{N}\|\bfr^\star\|_2^2 = (R_N^\star)^2, 
\ea 
where $(R_N^\star)^2 = (\Delta_N^\star)^2 + (Q_N^\star)^2$. Thus, the overall
approximation error $(R_N^\star)^2$ trades off the ability to explain between-cluster heterogeneity against the additional within-cluster leakage.

\begin{theorem} \label{thm0}
Let $\cfM^\star = \mbox{supp}(\bfbeta^\star)$ and $M^\star = |\cfM^\star|$. Suppose that the columns of $\bfX$ are normalized and that $\bfX$ satisfies a restricted eigenvalue (RE) condition on the relevant sparse cone.  When $\lambda \asymp \sqrt{\log(p)/N}$ and $P \left( \left\| N^{-1} \bfX^\top \bfepsilon \right\|_\infty \le \lambda/4 \right) \rightarrow 1$, we have
\begin{enumerate}
\item[(I)] 
\ba
\frac{1}{N} \lVert \bfX (\hat{\bfbeta} - \bfbeta^\star) \rVert_2^2 = O_p \left( (R_N^\star)^2 + \lambda^2 M^\star \right);
\ea
\item[(II)] If, in addition, $\left\| N^{-1} \bfX^\top \bfr^\star \right\|_\infty \le \lambda/4$, then
\ba
\frac{1}{N} \lVert \bfX (\hat{\bfbeta} - \bfbeta^\star) \rVert_2^2 = O_p \left( \lambda^2 M^\star \right), \quad 
\lVert \hat{\bfbeta} - \bfbeta^{\star} \rVert_1 = O_p \left( \lambda M^\star \right).
\ea   
\end{enumerate}
\end{theorem}
The proof is given in the supplementary material. The normalization, restricted eigenvalue and noise conditions are standard in the high-dimensional literature \citep{bickel2009simultaneous,wainwright2019high}; condition (C3) of Section \ref{sec:thm} extends the restricted eigenvalue condition to the augmented design used in a later section. Theorem \ref{thm0} gives oracle bounds centered at the contaminated coefficient vector $\bfbeta^\star$, rather than at the structural coefficient $\bfbeta^0$. Part (I) quantifies the effect of the overall surrogate approximation error $R_N^\star$. Part (II) shows that the usual sparse LASSO rates are recovered when the surrogate residual $\bfr^\star$ is weakly aligned with the design, in the sense that its contribution to the score is dominated by the regularization level. The following corollary formalizes the resulting \textit{target shift} when across-cluster heterogeneity in the covariate distributions provides sparse proxies for the latent cluster effects.

\begin{corollary}\label{cor0}
Under the conditions of Theorem \ref{thm0}(II), if $\lambda M^\star = o(1)$, then
\ba
\hat{\bfbeta} - \bfbeta^0 = \Tilde{\bfgamma}^\star + o_p(1).
\ea
\end{corollary}

Corollary \ref{cor0} explains the inflated false selection of the marginal-model LASSO in high-dimensional clustered data. When the high-dimensional covariates contain sufficient between-cluster information, a sparse linear combination of their cluster means can partially recover the latent cluster effects. In that case, (\ref{LASSO}) can predict well even though it is misspecified as a model for $\bfbeta^0$. However, because it targets $\bfbeta^\star = \bfbeta^0 + \Tilde{\bfgamma}^\star$ rather than $\bfbeta^0$, it tends to select variables that are predictive of the latent cluster effects rather than truly associated with the outcome through the structural effects. The shift is negligible when the between-cluster covariate components provide no useful sparse proxy for the latent cluster effects. In particular, when the covariate distributions are homogeneous  across clusters, the cluster means cannot explain variation in $\bfalpha$ and the marginal-model LASSO remains approximately centered at $\bfbeta^0$.

\subsection{The SHEL and GSHEL estimators} \label{subsec:SHEL}
To remove the target-shift, we introduce the SHEL estimator. The key idea is to separate the structural role of the individual-level covariates from their potential role as proxies for cluster heterogeneity, by explicitly constructing cluster-level summary statistics to absorb the between-cluster variation in $\bfalpha$.

Let $\widehat{\bfB}_c = \cfB(\bfX, \bfD) \in \mathbb{R}^{m \times p_0}$ and $\bfB = \bfD \widehat{\bfB}_c$, where the $i$th row of $\widehat{\bfB}_c$ is denoted by $\bfb_i^\top$. Thus, every column of $\bfB$ is constant within cluster. The map $\cfB(\cdot)$ can be heterogeneous cluster means, principal-component or sparse-principal-component scores, or other outcome-independent summaries of cluster-level covariate information. 

We therefore propose the following fixed-effects penalized estimator, or SHEL, to estimate $\bfbeta$ while using $\bfB \bfgamma$ as a synthetic approximation to $\bfalpha$. For the linear model (\ref{lmm}), SHEL is
\be \label{SHEL}
(\hat{\bfbeta}, \hat{\bfgamma}) = \underset{\bfbeta \in \mathbb{R}^{p}, \bfgamma \in \mathbb{R}^{p_0}}{\arg\min} \left\{ \frac{1}{2N} \lVert \bfY - \bfX \bfbeta - \bfB \bfgamma \rVert_2^2 + \lambda_1 \lVert \bfbeta \rVert_1 + \lambda_2 \lVert \bfgamma \rVert_1 \right\}.
\ee
For the generalized linear model (GLM), let $\ell(\eta; y) = b(\eta) - y \eta$ be the negative log-likelihood up to terms independent of the natural parameter $\eta$ in (\ref{glmm}). The GSHEL estimator is
\be \label{GSHEL}
(\hat{\bfbeta}, \hat{\bfgamma}) = \underset{\bfbeta \in \mathbb{R}^{p}, \bfgamma \in \mathbb{R}^{p_0}}{\arg\min} \left\{ \cfL_N(\bftheta) + \lambda_1 \lVert \bfbeta \rVert_1 + \lambda_2 \lVert \bfgamma \rVert_1 \right\}, \quad \cfL_N(\bftheta) = \frac{1}{N} \sum_{i,j} \ell(\bfW_{ij}^\top \bftheta; y_{ij}),
\ee
where $\bfW = [\bfX ~ \bfB]$, $\bfW_{ij} = (\bfX_{ij}^\top, \bfb_i^\top)^\top$, and $\bftheta = (\bfbeta^\top, \bfgamma^\top)^\top$. SHEL is therefore the Gaussian identity-link special case of GSHEL. 

\subsection{Within- and between-cluster effects and identifiability} \label{subsec:interpret}
To clarify the structural interpretation of the proposed adjustment, consider the case when $\bfB$ consists of the empirical cluster means of covariates, then 
\be \label{within_between}
\bfX_{ij}^\top \bfbeta + \bar{\bfX}_i^\top \bfgamma = \bfxi_{ij}^\top \bfbeta + \bar{\bfX}_i^\top (\bfbeta + \bfgamma).
\ee
Thus, $\bfbeta$ is identified through individual-level variation around the cluster means, while the cluster-level effects are absorbed jointly by $\bfbeta + \bfgamma$. This is closely related to the within-between decomposition for clustered and panel data \citep{mundlak1978pooling,neuhaus1998between}. 

In contrast, the marginal-model LASSO uses
\ba 
\bfX_{ij}^{\top} \bfbeta = \bfxi_{ij}^{\top} \bfbeta + \bar{\bfX}_i^{\top} \bfbeta, 
\ea
and so constrains the within- and between-cluster coefficients to be identical. When $\alpha_i$ is associated with $\bar{\bfX}_i$, the between-cluster component generally favors a coefficient different from $\bfbeta^0$. The marginal-model LASSO can accommodate this discrepancy only by shifting the common coefficient away from $\bfbeta^0$ to $\bfbeta^\star$. The same shift also enters the within-cluster component through $q_{ij}^\star$, which produces the additional leakage term $(Q_N^\star)^2$ in Theorem \ref{thm0}.

The role of $\bfB$ is therefore not to recover the heterogeneous intercepts, but to provide a structured correction to the latent cluster-level heterogeneity that might otherwise be spuriously absorbed by $\bfX$. We require only the existence of at least one sufficiently low-complexity synthetic representation. Let $\bfgamma^0$ be such a representation and write
\be \label{alpha_decomp}
\alpha_i = \bar{\bfX}_{i,\cfP_0}^{\top} \bfgamma^0 + \delta_i, \quad i = 1,\ldots,m.
\ee
Thus, the between-cluster component of the true linear predictor can be written as
\ba
\bar{\bfX}_{i,\cfP_0}^{\top} \bfbeta_{\cfP_0}^0 + \alpha_i = \bar{\bfX}_{i,\cfP_0}^{\top} \bfbeta_{\cfP_0}^0 + \bfgamma^0) + \delta_i.
\ea 
For any such representation, let $\bftheta^0 = \left((\bfbeta^0)^\top, (\bfgamma^0)^\top\right)^\top$. The true natural parameter therefore satisfies $\eta_{ij}^0 = \bfW_{ij}^\top \bftheta^0 + \delta_i$. We define $\Delta_N^2 = N^{-1} \| \bfD \bfdelta \|_2^2 = \sum_{i=1}^m \pi_i \delta_i^2$ as the approximation error.

This decomposition also clarifies the identifiability condition. For a target coordinate $l$,  sufficient variation in $X_{ij,l}$ must remain after adjusting for the remaining covariates and the synthetic covariates. In particular, a cluster-constant covariate cannot be separated from an unrestricted cluster-level nuisance effect without additional structure. A low-dimensional population-level condition can be expressed as
\ba
\Var \left[ X_{ij,l} - \E \{ X_{ij,l} \mid \bfX_{ij,-l}, \bfb_i \} \right] \ge c_l > 0.
\ea
In high dimensions, it is represented by the RE condition (C3) in Section \ref{sec:thm} for the augmented design and by the nondegeneracy of the inverse-Hessian direction used in Section \ref{sec:ps}. 

\section{Theoretical properties} \label{sec:thm}
\subsection{Regularity conditions} \label{subsec:condition}
Throughout this section, the columns of $\bfW$ are normalized such that $N^{-1} \| \bfW_k \|_2^2 \le 1$ for every $k$. Let $\cfM_2 = \mbox{supp}(\bfgamma^0)$ and $M_2 = |\cfM_2|$, and assume that the penalty levels satisfy 
\ba
\lambda_1 \asymp \sqrt{\frac{\log p}{N}}, \quad \lambda_2 \asymp \sqrt{\frac{\log p_0}{N}}, \quad M_1 \lambda_1^2 + M_2 \lambda_2^2 = o(1).
\ea
Define $\bfg_N^0 = \E \{ \nabla \cfL_N(\bftheta^0) \mid \bfX, \bfD, \bfalpha \}$, partitioned as $\bfg_N^0 = ( (\bfg_{\beta,N}^0)^\top, (\bfg_{\gamma,N}^0)^\top )^\top$. We impose the following regularity conditions.

\begin{description}
\item[(C1)] \textit{Outcome-independent synthetic construction.} The synthetic dictionary $\widehat{\bfB}_c = \cfB(\bfX, \bfD)$ is measurable with respect to $\cfG_X = \sigma(\bfX, \bfD)$.
\item[(C2)] \textit{Sparse synthetic approximation (SSA).} There exists $\bfgamma^0 \in \mathbb{R}^{p_0}$ satisfying (\ref{alpha_decomp}) with $\|\bfgamma^0\|_0 \le M_2$. 
\item[(C3)] {\it Weighted restricted eigenvalue condition.} For some fixed constant $a > 1$, let
\ba
\cfC_a(\cfM_1, \cfM_2) = \left\{ \bfu: \lambda_1 \| \bfu_{1,\cfM_1^c} \|_1 + \lambda_2 \|\bfu_{2, \cfM_2^c} \|_1 \le a \left( \lambda_1 \| \bfu_{1,\cfM_1} \|_1 + \lambda_2 \| \bfu_{2,\cfM_2} \|_1\right) \right\}.
\ea
There exist constants $\kappa>0$ and $r>0$ such that, with probability tending to one,
\ba
\frac{1}{N} \sum_{i=1}^m \sum_{j=1}^{n_i} b''(\bfW_{ij}^\top \bftheta) (\bfW_{ij}^\top \bfu)^2 \ge \kappa \|\bfu\|_2^2
\ea
for all $\bfu \in \cfC_a(\cfM_1, \cfM_2)$ and $\|\bftheta - \bftheta^0\|_1 \le r$.
\item[(C4)] \textit{Centered score concentration.} For some sufficiently small $c_s > 0$,
\ba
\| \nabla_{\bfbeta} \cfL_N(\bftheta^0) - \bfg_{\beta,N}^0 \|_\infty \le c_s \lambda_1,
\quad
\| \nabla_{\bfgamma} \cfL_N(\bftheta^0) -\bfg_{\gamma,N}^0 \|_\infty \le c_s \lambda_2
\ea
with probability tending to one.
\item[(C5)] \textit{Structural-score condition.} For a constant $c_b>0$ satisfying $c_s + c_b < 1/2$,
\ba
\| \bfg_{\beta,N}^0\|_\infty \le c_b \lambda_1, \quad \| \bfg_{\gamma,N}^0 \|_\infty\le c_b \lambda_2.
\ea
\end{description}
Condition (C1) allows $\bfB$ to be constructed adaptively from the covariates but not from the response, and (C2) requires that the cluster intercepts admit a sparse representation in the synthetic dictionary up to the residual $\delta_i$. The SSA is analogous to the approximate sparsity assumption commonly used in high-dimensional regression and semiparametric inference \citep{bickel2009simultaneous,bunea2007sparsity,belloni2014pivotal}. Related assumptions also appear in work on factor-LASSO \citep{hansen2019factor} and proxy controls \citep{deaner2021many}, where latent factors or confounders are represented via high-dimensional observed proxies.  Condition (C3) extends the RE condition of Theorem \ref{thm0} to the Fisher-weighted augmented design. For the Gaussian identity-link model, $b''(\cdot) = 1$, and (C3) reduces to the usual RE condition for $\bfW$, with the cone modified to account for the two penalty levels. Condition (C4) controls the stochastic component of the score, whereas (C5), which is used only in Corollary \ref{cor2} for a sharper consistency result, controls the remaining population-score distortion. For a canonical GLM,
\ba
\bfg_N^0 = \frac{1}{N} \sum_{i=1}^m\sum_{j=1}^{n_i} \bfW_{ij} \left[ b'(\bfW_{ij}^\top \bftheta^0) - b'(\bfW_{ij}^\top \bftheta^0 + \delta_i) \right].
\ea
If $b''(\cdot)$ is uniformly bounded, we have $\|\bfg_N^0\|_\infty \le C \Delta_N$, and $\Delta_N = O\{ \min(\lambda_1, \lambda_2) \}$ is sufficient for (C5). The score condition itself can nevertheless hold under weaker requirements when the residual heterogeneity is weakly aligned with the augmented design.

\subsection{Oracle properties} \label{subsec:SHEL_thm}
This section provides theoretical properties of the GSHEL estimator. Proofs are given in the supplementary material. For any reference parameter $\bftheta^\circ$, we define the excess empirical loss as
\ba
\cfD_N(\bftheta, \bftheta^\circ) = \cfL_N(\bftheta) - \cfL_N(\bftheta^\circ) - \nabla \cfL_N(\bftheta^\circ)^\top (\bftheta - \bftheta^\circ).
\ea 

\begin{theorem} \label{thm1}
Under (C1)-(C4), when $b''(\cdot)$ is bounded above and away from 0 on the linear predictors $\bfW_{ij}^\top \bftheta$ with $\| \bftheta - \bftheta^0 \|_1 \le r$ and on the segments between $\bfW_{ij}^\top \bftheta^0$ and $\bfW_{ij}^\top \bftheta^0 + \delta_i$. Then
\begin{enumerate}
\item[(I)] 
\ba
\cfD_N( \hat{\bftheta}, \bftheta^0 ) = O_p \left( \Delta_N^2 + M_1 \lambda_1^2 + M_2 \lambda_2^2 \right).
\ea
\item[(II)] 
\ba
\lambda_1 \| \hat{\bfbeta} - \bfbeta^0 \|_1 + \lambda_2 \| \hat{\bfgamma} - \bfgamma^0 \|_1 = O_p \left( \Delta_N^2 + M_1 \lambda_1^2 + M_2 \lambda_2^2 \right).
\ea 
\end{enumerate}
Consequently,
\ba
\| \hat{\bfbeta} - \bfbeta^0 \|_1 = O_p \left( M_1 \lambda_1 + M_2\frac{\lambda_2^2}{\lambda_1} + \frac{\Delta_N^2}{\lambda_1} \right).
\ea
\end{theorem}
Theorem~\ref{thm1} separates the estimation error into two components. The term $M_1 \lambda_1^2 + M_2 \lambda_2^2$ is the usual high-dimensional estimation error associated with the two penalized blocks, whereas $\Delta_N^2$ measures the approximation error remaining after the cluster-level heterogeneity is represented through the synthetic dictionary. The curvature condition of the theorem is automatic for the Gaussian identity-link model and, for example, holds for logistic regression when the relevant linear predictors are uniformly bounded.

In the linear model (\ref{lmm}), $\bfg_N^0 = -N^{-1} \bfW^\top \bfD \bfdelta$ and $\cfD_N(\hat{\bftheta}, \bftheta^0) = (2N)^{-1} \|\bfW(\hat{\bftheta} - \bftheta^0)\|_2^2$, so Theorem \ref{thm1} yields a bound on the prediction error that can be compared directly with Theorem \ref{thm0}.

\begin{proposition} \label{prop1}
Under model (\ref{lmm}) and conditions (C1)-(C4), 
\ba
\frac{1}{N} \| \bfW ( \hat{\bftheta} - \bftheta^0 ) \|_2^2 = O_p \left( \Delta_N^2 + M_1 \lambda_1^2 + M_2 \lambda_2^2 \right).
\ea
\end{proposition}

Proposition \ref{prop1} and Theorem \ref{thm0} both decompose the prediction error into an approximation component and a high-dimensional estimation component. For the marginal-model LASSO, however, the approximation error is $(\Delta_N^\star)^2 + (Q_N^\star)^2$, and the estimator concentrates around a contaminated parameter $\bfbeta^\star$ rather than the structural parameter $\bfbeta^0$. The additional term $(Q_N^\star)^2$ arises because individual-level covariates are used to approximate a cluster-constant effect, thereby introducing within-cluster leakage. In contrast, SHEL represents the nuisance effect through the cluster-constant dictionary $\bfB$, so Proposition \ref{prop1} contains only the residual synthetic-approximation error $\Delta_N^2$ and remains centered at the structural parameter $\bftheta^0$. The comparison does not imply that SHEL always has smaller error, since its performance also depends on the approximation quality and complexity of the chosen synthetic dictionary, but it identifies the specific source of error avoided by cluster-constant nuisance adjustment.

Condition (C5) describes a regime in which the residual heterogeneity has a negligible effect on the score, and therefore yields a sharper consistency result. 

\begin{corollary} \label{cor2}
Under conditions (C1)-(C5),
\begin{enumerate}
\item[(I)] 
\ba
\cfD_N( \hat{\bftheta}, \bftheta^0 ) = O_p \left( M_1 \lambda_1^2 + M_2 \lambda_2^2 \right).
\ea
\item[(II)] 
\ba
\lambda_1 \| \hat{\bfbeta} - \bfbeta^0 \|_1 + \lambda_2 \| \hat{\bfgamma} - \bfgamma^0 \|_1 = O_p \left( M_1 \lambda_1^2 + M_2 \lambda_2^2 \right).
\ea 
\end{enumerate}
Consequently,
\ba
\| \hat{\bfbeta} - \bfbeta^0 \|_1 = O_p \left( M_1 \lambda_1 + M_2\frac{\lambda_2^2}{\lambda_1} \right).
\ea
\end{corollary}

Corollary~\ref{cor2} gives the usual sparse high-dimensional rate without an explicit approximation-error term. The key requirement is that the $\Delta_N$-induced population-score distortion be dominated by the corresponding penalty levels. That is, even when $\Delta_N$ is relatively large, (C5) can still hold provided that the remaining cluster heterogeneity is sufficiently orthogonal to the relevant directions in $\bfW$.

\subsection{Marginal working parameter under persistent residual heterogeneity} \label{subsec:pseudo}
When the dictionary leaves substantial residual heterogeneity that is aligned with the design, (C5) can fail, and for nonlinear links, the estimand itself changes. To make this explicit, we adopt a superpopulation formulation in this subsection and in Section \ref{sec:ps}, and let $\cfL (\bftheta) = \E \{ \ell(\bfW_{ij}^\top \bftheta; Y_{ij}) \}$ to be the population risk, with the corresponding {\it marginal-model parameter} defined by
\be \label{marginal}
\Tilde{\bftheta} = \underset{\bftheta \in \mathbb{R}^{p+p_0}}{\arg\min} ~ \cfL (\bftheta).
\ee
Equivalently, $\Tilde{\bftheta}$ solves the population score equation $\E [\nabla \cfL (\bftheta)] = \bfzero$. 

Under the superpopulation formulation with independent clusters, the corresponding marginal pseudo-parameter is the minimizer of the expected cluster loss. The notation $\widetilde{\bftheta}$ is used for the limiting population quantity.

The distinction between $\bftheta^0$ and $\widetilde{\bftheta}$ is intrinsic to nonlinear marginalization. Even when a stochastic residual effect satisfies $\E(\delta_i \mid \bfW_i) = 0$, in general,
\ba
\E \{ b'(\bfW_{ij}^\top \bftheta^0 + \delta_i ) \mid \bfW_i \} \ne b'(\bfW_{ij}^\top \bftheta^0),
\ea
so the marginal working score need not vanish at the structural coefficient. 

Let $\bfg(\bftheta) = \nabla \cfL (\bftheta)$ and define the integrated Hessian
\ba
\bar{\bfH} = \int_0^1 \nabla^2 \cfL \left\{ \bftheta^0 + t(\widetilde{\bftheta} - \bftheta^0) \right\} \,dt.
\ea
When $\bar{\bfH}_N$ is nonsingular, the fundamental theorem of calculus gives the identity
\ba
\widetilde{\bftheta}_N - \bftheta^0 = -\bar{\bfH}_N^{-1} \bfg_N^0.
\ea
Thus the same population-score discrepancy that determines structural consistency also determines the difference between the structural and marginal working targets.

\begin{proposition} \label{prop2}
Suppose $\bar{\bfH}$ is nonsingular and, for a target coordinate $l$,
\ba
\|\bfe_l^\top \bar{\bfH}^{-1}\|_1 \le C_H < \infty.
\ea
Then
\ba
|\widetilde \beta_{l} - \beta_l^0| \le C_H \|\bfg_N^0\|_\infty.
\ea
If, in addition, $b''(\cdot)$ is uniformly bounded and the augmented design is normalized as in (C3), then
\ba
|\widetilde \beta_{l} - \beta_l^0| = O(\Delta_N).
\ea    
\end{proposition}
Proposition \ref{prop2} bounds the shift coordinatewise, and it implies that a consistent estimator of $\Tilde{\bfbeta}$ can retain the sparsity pattern of $\bfbeta^0$. One example is the logistic regression model with $\delta_i \sim N(0, u^2)$. Using the cumulative Gaussian approximation to the logistic function, \citep{zeger1988models} showed that
\ba
\widetilde{\beta}_{l} \approx \frac{\beta_l^0}{\sqrt{1 + c^2 u^2}}, \quad c = \frac{16 \sqrt{3}}{15 \pi}.
\ea
For the linear model (\ref{lmm}), when $\bfH = \E \left( \bfW_{ij} \bfW_{ij}^\top \right)$ is nonsingular, then $\widetilde{\bftheta} - \bftheta^0 = \bfH^{-1} \E \left( \bfW_{ij} \delta_i \right)$, which equals $\bfzero$ when $\E(\delta_i | \bfW_{ij}) = 0$.

The next theorem shows that the GSHEL estimator is consistent around the marginal-model target $\widetilde \bftheta$ under conditions analogous to (C3) and (C4) at $\widetilde \bftheta$. Let $\widetilde{\cfM}_1 = \mbox{supp}(\widetilde{\bfbeta})$, $\widetilde{\cfM}_2 = \mbox{supp}(\widetilde{\bfgamma})$, $\widetilde M_1 = |\widetilde{\cfM}_1|$ and $\widetilde M_2 = |\widetilde{\cfM}_2|$. 

\begin{theorem} \label{thm2}
Suppose that the clusters are independent and identically distributed with $\max_i n_i \le C_n < \infty$, $\widetilde M_1 \lambda_1^2 + \widetilde M_2 \lambda_2^2 = o(1)$, the weighted RE condition in (C3) holds with $(\bftheta^0, \cfM)$ replaced by $(\widetilde{\bftheta}, \widetilde{\cfM})$, and
\ba
\| \nabla_{\bfbeta} \cfL_N(\widetilde{\bftheta}) \|_\infty \le c_s \lambda_1,
\quad
\| \nabla_{\bfgamma} \cfL_N(\widetilde{\bftheta}) \|_\infty \le c_s \lambda_2
\ea
with probability tending to one. Then
\ba
\cfD_N( \hat{\bftheta}, \widetilde{\bftheta} ) = O_p \left( \widetilde M_1 \lambda_1^2 + \widetilde M_2 \lambda_2^2 \right)
\ea
and
\ba
\lambda_1 \| \hat{\bfbeta} - \widetilde{\bfbeta} \|_1 + \lambda_2 \| \hat{\bfgamma} - \widetilde{\bfgamma} \|_1 = O_p \left( \widetilde M_1 \lambda_1^2 + \widetilde M_2 \lambda_2^2 \right).
\ea
\end{theorem}
Unlike Corollary \ref{cor2}, Theorem \ref{thm2} does not require the structural-score condition (C5), because the population score vanishes at $\widetilde{\bftheta}$ by definition under the superpopulation formulation. 


\section{Debiased inference} \label{sec:ps}
We next develop coordinatewise inference for the marginal-model working parameter $\widetilde{\bftheta}$. Structural inference for $\beta_l^0$ follows when $\widetilde \beta_l$ and $\beta_l^0$ are asymptotically equivalent at the $m^{-1/2}$ scale. The construction follows the inverse-Hessian debiasing principle for high-dimensional regression \citep{zhang2014confidence, van2014asymptotically}, with cluster-level asymptotic variance. 

Let $\widetilde{\bfH} = \nabla^2 \cfL(\widetilde{\bftheta})$ denote the population Hessian matrix. For a target index $l \in \{ 1, \cdots, p \}$, let $\bfa_l^0 = \widetilde{\bfH}^{-1}_l$ denote the $l$th column of $\widetilde{\bfH}^{-1}$. By inverting the KKT condition, \cite{van2014asymptotically} proposed the one-step bias-corrected estimator
\be \label{debias}
\hat{\beta}_l^d = \hat{\beta}_l - \hat{\bfa}_l^\top \nabla \cfL_N(\hat{\bftheta}) = \hat{\beta}_l + \hat{\bfa}_l^\top \frac{1}{N} \sum_{i,j} \bfW_{ij} (y_{ij} - \hat{\mu}_{ij}),
\ee
where $\hat{\bfa}_l$ is an estimator of $\bfa_l^0$. The purpose of the correction term is to remove the first-order shrinkage bias induced by the $L_1$ penalty.

To estimate $\bfa_l^0$, one commonly assumes that the precision matrix $\widetilde{\bfH}^{-1}$ is sparse and uses a nodewise regression argument \citep{meinshausen2006high,van2014asymptotically}. Specifically, let $\bfZ = \widehat{\bfV}^{1/2} \bfW$ be the Fisher-weighted augmented design, where $\widehat{\bfV}$ have diagonal entries $\hat v_{ij} = b''(\bfW_{ij}^\top \hat{\bftheta})$. We consider the nodewise regression \citep{van2014asymptotically}, where
\ba
\hat{\bfa}_l = \hat \tau_l^{-2} (1, ~- \hat{\bfpi}_l^\top)^\top, \quad \hat \tau_l^2 = \frac{1}{N} \bfZ_l^\top \hat{\bfr}_l, \quad \hat{\bfr}_l = \bfZ_l-\bfZ_{-l}\hat{\bfpi}_l, 
\ea
with 
\ba
\hat{\bfpi}_l = \underset{\bfpi}{\arg\min} \left\{ \frac{1}{2N} \|\bfZ_l - \bfZ_{-l}\bfpi\|_2^2 + \lambda_l^{\rm node} \|\bfpi\|_1 \right\},
\quad \lambda_l^{\rm node} \asymp \sqrt{\frac{\log(p+p_0)}{N}}.
\ea
At the population level, the debiasing step requires $\bfSigma_\delta \bfa_l^0 = \bfe_l$, so that the score
\ba
\bfa_l^{0,\top}\frac{1}{N}\sum_{i,j}\bfW_{ij}(y_{ij}-\mu_{ij})
\ea
is orthogonal to the nuisance directions. The nodewise regression step estimates $\bfa_l^0$ through a sparse approximation to the linear projection of $Z_l$ onto the column space of $\bfZ_{-l}$. When $\bfZ$ is multivariate Gaussian, this linear projection is exact, and $\bfa_l^0$ reflects the corresponding conditional independence structure \citep{ye2026high}. For non-Gaussian $\bfZ$, the linear projection remains well defined whenever the second moment of $Z_l$ exists.

Define the oracle observation- and cluster-level influence functions
\ba
\phi_{l,ij} = -(\bfa_l^0)^\top \nabla \ell( \bfW_{ij}^\top \widetilde{\bftheta}; Y_{ij} ), \quad
\Phi_{l,i} = \frac{m}{N} \sum_{j=1}^{n_i }\phi_{l,ij},
\ea
with plug-in estimates
\ba
\hat \phi_{l,ij} = -\hat{\bfa}_l^\top \nabla \ell( \bfW_{ij}^\top \hat{\bftheta}; Y_{ij} ), \quad
\hat \Phi_{l,i} = \frac{m}{N} \sum_{j=1}^{n_i} \hat \phi_{l,ij}.
\ea
Let
\be \label{cluster_var}
\hat V_l = \frac{1}{m} \sum_{i=1}^m (\hat \Phi_{l,i} - \bar \Phi_l)^2, \quad
\bar \Phi_l = \frac{1}{m} \sum_{i=1}^m \hat \Phi_{l,i}.
\ee
Then $\widehat{\rm se}(\hat \beta_l^{d}) = \sqrt{\hat V_l / m}$. We impose the following additional conditions for inference:

\begin{description}
\item[(I1)] The clusters are independent and identically distributed, $\max_i n_i \le C_n < \infty$, and $\E |\Phi_{l,i}|^{2 + \nu} < \infty$ for some $\nu > 0$.
\item[(I2)] The initial estimator satisfies the rates in Theorem~\ref{thm2}; in particular,
\ba
\lambda_1 \| \hat{\bfbeta} - \widetilde{\bfbeta} \|_1 + \lambda_2 \| \hat{\bfgamma} - \widetilde{\bfgamma} \|_1 = O_p \left( \widetilde M_1 \lambda_1^2 + \widetilde M_2 \lambda_2^2 \right).
\ea
\item[(I3)] For some $M_{a,l}$,
\ba
\|\hat{\bfa}_l - \bfa_l^0\|_1 = O_p(M_{a,l} \lambda_l^{\rm node}), \quad
\|\bfe_l - \widehat{\bfH} \hat{\bfa}_l\|_\infty = O_p(\lambda_l^{\rm node}),
\ea
where $\widehat{\bfH} = N^{-1} \bfW^\top \widehat{\bfV} \bfW$.
\item[(I4)] Let $\bfr_l = \bfe_l - \widehat{\bfH} \hat{\bfa}_l
= (\bfr_{l,1}^\top, \bfr_{l,2}^\top)^\top$, where $\bfr_{l,1} \in \reals^p$ and $\bfr_{l,2} \in \reals^{p_0}$. The debiasing remainders satisfy
\ba
&~& \sqrt m \left\{ \|\bfr_{l,1}\|_\infty \| \hat{\bfbeta} - \widetilde{\bfbeta} \|_1 + \|\bfr_{l,2}\|_\infty \| \hat{\bfgamma} - \widetilde{\bfgamma} \|_1 \right\} = o_p(1), \\
&~& \sqrt m\,
\| \hat{\bfa}_l - \bfa_l^0 \|_1 \| \nabla \cfL_N( \widetilde{\bftheta} )\|_\infty = o_p(1), \quad \sqrt m \left| \hat{\bfa}_l^\top (\widehat{\bfH} - \bar{\bfH}_N) ( \hat{\bftheta} - \widetilde{\bftheta} ) \right| = o_p(1)
\ea
where
\ba
\bar{\bfH}_N = \int_0^1 \nabla^2 \cfL_N \left\{ \widetilde{\bftheta} + t ( \hat{\bftheta} - \widetilde{\bftheta} ) \right\} dt.
\ea
\item[(I5)] The variance $V_l$ satisfies $0 < c_V \le V_l \le C_V < \infty$, and $m^{-1} \sum_{i=1}^m ( \hat \Phi_{l,i} - \Phi_{l,i} )^2 = o_p(1)$.
\end{description}
Conditions (I1) and (I5) provide moment conditions for a cluster-level central limit theorem. Conditions (I2)-(I4) control the initial-estimation, inverse-Hessian, and nonlinear Taylor remainders. In particular, combined with Theorem \ref{thm2} and the nodewise KKT bound, a simple sufficient rate for the first display in (I4) is
\ba
\sqrt m \, \lambda_l^{\rm node} \left( \widetilde M_1 \lambda_1^2 +  \widetilde M_2 \lambda_2^2 \right) \left( \frac{1}{\lambda_1} + \frac{1}{\lambda_2} \right) \rightarrow 0.
\ea

\begin{theorem} \label{thm4}
For any $l \in \{ 1, \ldots, p \}$, under (I1)-(I5),
\begin{enumerate}
\item[(I)] 
\ba
&~& \sqrt m( \hat \beta_l^{d} - \widetilde \beta_l) = \frac{1}{\sqrt m} \sum_{i=1}^m \Phi_{l,i} + o_p(1).
\ea
\item[(II)]  $\hat V_l \xrightarrow{P} V_l$ and 
\ba
\frac{\sqrt m ( \hat \beta_l^{d} - \widetilde \beta_l )}{\sqrt{\hat V_l}} \xrightarrow{d} N(0,1).
\ea
\end{enumerate}
\end{theorem}

Theorem \ref{thm4} gives coordinatewise inference for a prespecified coordinate of the marginal working parameter. If the reported coordinates are selected data-adaptively, the selection step should be carried out on an independent set of clusters. An advantage of the debiased inference is that its validity does not require a sure-screening property \citep{fan2008sis}.

\begin{corollary} \label{cor_structural_inference}
Under the conditions of Theorem~\ref{thm4}, if

$$
|\widetilde{\beta}_l-\beta_l^0|=o(m^{-1/2}),
$$

then

$$
\frac{\sqrt m(\hat\beta_l^d-\beta_l^0)}
{\sqrt{\hat V_l}}
\xrightarrow{d}N(0,1).
$$

\end{corollary}

Thus, inference for the population working parameter also yields valid inference for the structural coefficient whenever the difference between the two targets is negligible at the cluster-level inference scale. By Proposition~\ref{prop2}, $\Delta_N=o(m^{-1/2})$ provides one sufficient condition under its stated regularity conditions.

\section{Simulation studies} \label{sec:sim}
This section reports the finite-sample performance of the proposed procedure. Implementation details of GSHEL, and an iterative algorithm, IGSHEL, are provided in the supplementary material.

\subsection{Simulation settings} \label{subsec:sim0}
We consider $m = 400$ clusters of equal size $n = n_i = 4$ and $p = 1000$ covariates. The heterogeneous intercepts $\alpha_i$ are randomly sampled from $F_{\alpha}$. Given the cluster, the covariates $\bfX_{ij}$ are drawn from a multivariate Gaussian distribution with mean vector $\bfmu_i$ and precision matrix $\bfTheta$, where $\bfTheta$ is block diagonal with block size $5$. Within each block, we set $\Theta_{ll} = 1$ and $\Theta_{ll'} = 0.5$ for $l \ne l'$. For each replication, we randomly select $p_0 = 0, 50, 100, 200, 500,$ or $800$ covariates to have heterogeneous distributions across clusters. That is, for $l \in \cfP_0^c$, we set $\mu_{l,i} = 0$ for all $i$; for $l \in \cfP_0$, the cluster-specific means $\mu_{l,i}$ are independently sampled from $F_l$.

We consider two specifications for $F_{\alpha}$ and $F_l$. The first, denoted by ``Independent", assumes that $F_{\alpha}$ and $F_l$ are independent for all $l$. In this setting, we let each $F_l$ be a standard Gaussian distribution, and let $F_{\alpha}$ follow either a standard Gaussian distribution or a Gaussian mixture, ($0.5N(-1, ~0.5) + 0.5N(1, ~0.5)$). In the second setting, denoted by ``Endogenous", we first generate a latent variable $U_i$, and then generate $\mu_{l,i}$ and $\alpha_i$ according to
\ba
\mu_{l,i} = h_l U_i + n^{-1} Z_{l,i}, ~~ \alpha_i = 0.8 U_i + 0.2 Z_{i}, 
\ea
where $Z_{l, i}$ and $Z_i$ follow standard Gaussian distributions and $h_l \sim \text{U}(0,1)$. As in the Independent setting, we let $U_i$ follow either a standard Gaussian or a Gaussian mixture.

For variable selection, parameter estimation, and prediction, we compare the proposed GSHEL and IGSHEL algorithms with the marginal-model LASSO estimator. For all methods, the tuning parameter is selected by ten-fold cross-validation (CV). We choose $\lambda$ either by minimizing the CV error ($\lambda_{\min}$) or by using the one-standard-error rule ($\lambda_{\text{1se}}$) \citep{hastie2009elements}. In the iterative procedure, ``ISHEL1'' uses $\lambda_1 = \lambda_{\text{1se}}$ in Step 6 of Algorithm 2, whereas ``ISHEL2'' uses $\lambda_1 = \lambda_{\min}$. We consider both the linear model (\ref{lmm}) and the logistic model for (\ref{glmm}). We set $\cfM_1 = \{ 1,6,11,12,16,17 \}$ and $\bfbeta_{\cfM_1} = (0.5, 0.5, 1, 1, 1.5, 1.5)^\top$. To evaluate variable selection performance, we record the number of false positives (FP) and true positives (TP). For the linear model, parameter estimation and prediction performance are assessed using the root mean squared error (RMSE), $\lVert \hat{\bfbeta} - \bfbeta^0 \rVert_1$, and the residual ICC. For the logistic model, prediction performance is evaluated using sensitivity and specificity. 

For inference, we consider a weaker signal vector $\bfbeta_{\cfM_1} = (0.25, 0.25, 0.40, 0.40, 0.60, 0.60)^\top$ and generate independent screening and inference samples, each containing 200 clusters. The screening sample determines the coordinates to be tested, while GSHEL estimation, debiasing, and interval construction use only the inference sample. As a benchmark, we fit a GLMM via unpenalized maximum likelihood estimation to the same inference sample. We report the empirical type I error rate, power, and median length of the 95\% confidence interval. The definitions of these measures are given in Table S1. All performance measures are computed from 200 simulated datasets.

\subsection{Simulation results for variable selection, parameter estimation, and prediction} \label{subsec:sim1}
The overall pattern is that the proposed estimators reduce the two measures of target shift, the number of false positives and $\lVert \hat{\bfbeta} - \bfbeta^0 \rVert_1$, to the level of the benchmark case $p_0 = 0$ without loss of prediction accuracy. Across all settings, all methods are able to identify all nonzero coefficients, yielding TP $= 6$. For linear models with Gaussian random effects, the numbers of false positives, RMSE, $\lVert \hat{\bfbeta} - \bfbeta^0 \rVert_1$, and residual ICC are presented in Figures \ref{Fig_FP} and \ref{Fig_e}. As expected, $\lambda_{\text{1se}}$ generally leads to better variable selection, whereas $\lambda_{\min}$ tends to yield better prediction.

When $p_0 = 0$, as suggested in Theorem \ref{thm0}, no covariates are available to synthesize the random intercepts, and the marginal model LASSO estimator nearly coincides with SHEL. As reflected by the small values of $\lVert \hat{\bfbeta} - \bfbeta^0 \rVert_1$ and the number of false positives, the estimation target $\bfbeta^0$ is consistently estimated, whereas RMSE and residual ICC are larger than cases with $p_0 > 0$ because the omitted random intercepts are absorbed into the residual variation.

As $p_0$ increases, the marginal model LASSO estimator tends to include more noise variables, resulting in larger FP and estimation error $\lVert \hat{\bfbeta} - \bfbeta^0 \rVert_1$. This degradation is due to contamination of $\bfgamma^\star$ and target shift, as suggested in Theorem \ref{thm0}. When comparing the two dependence settings, the marginal-model LASSO is less sensitive to increases in $p_0$ under the Independent setting than under the Endogenous setting. Under the Independent setting, the synthetic approximation is driven by high-dimensional spurious correlations among variables in $\cfP_0$, rather than by genuine dependence between the covariates and the random intercepts. In contrast, our proposed method performs similarly to the benchmark case $p_0 = 0$.

For both the proposed methods and the marginal-model LASSO, RMSE and residual ICC decrease as $p_0$ increases, reflecting the increasing ability of the synthetic approximation to explain cluster-level heterogeneity. SHEL and ISHEL perform similarly under the Endogenous setting, whereas ISHEL outperforms SHEL under the Independent setting. An explanation is that, under Independence, the optimal coefficient vector $\bfgamma^0$ is generally not sparse, as the synthetic approximation relies on accumulated spurious high-dimensional correlations, and the iterative procedure is better suited to accommodate this structure.

For the logistic model with Gaussian random effects, the numbers of false positives, sensitivity, and specificity are presented in Figures S1 and S2. The conclusions are similar to those for the linear model, although the marginal-model LASSO appears to be less affected by the target shift, presumably because a binary response carries less information about the random intercepts than a continuous one. In addition, our proposed methods show similar performance under Gaussian mixture random effects (see Figures S3-S6), implying the robustness to non-Gaussian deviations in the random effects distribution.

\subsection{Simulation results for post-selection inference} \label{subsec:sim2}
Table \ref{tab:debias} summarizes the performance of the debiased procedure under Gaussian cluster effects. The proposed debiased GSHEL procedure provided stable type I error control across both the Endogenous and Independent settings. For the linear model, the pooled type I error remained close to the nominal 0.05 level across all values of $p_0$, while power was one despite the relatively weak signal setting. For the logistic model, pooled type I error was similarly stable, and power remained above 0.94 across all settings.

The behavior of the unpenalized mixed-effects refit differed substantially between the two dependence structures. Under the Independent setting, the GLMM generally maintained reasonable type I error control. Under the Endogenous setting, however, its type I error increased markedly as the number of covariates with heterogeneous cluster means increased. 

These results are consistent with the identification argument developed in Section \ref{subsec:interpret}. When the latent cluster effect is independent of the covariates, a conventional random-effects model can recover the structural effects without explicitly modeling the dependence between the cluster effect and the covariates. Under endogenous heterogeneity, this exogeneity condition is violated, resulting in distorted coefficient estimates and anti-conservative inference from the low-dimensional mixed-effects refit. Similar performance is observed under Gaussian mixture random effects; see Table S2. 

\section{Application to longitudinal neutrophil transcriptomics in COVID-19} \label{sec:real}
\subsection{Background and data description}
We applied GSHEL to a publicly available longitudinal bulk RNA-seq dataset of enriched blood neutrophils from hospitalized COVID-19 patients \citep{lasalle2022longitudinal}, available from GEO (GSE212041). The original study enrolled 306 COVID-19-positive patients at Massachusetts General Hospital between March and May 2020. Blood samples were collected on day 0, day 3, and day 7 after enrollment, when available, and neutrophils were enriched from whole blood for bulk RNA sequencing. After quality control, 629 samples from 306 patients were available for analysis. Gene expression was quantified as $\log_2(\mathrm{TPM}+1)$, and the binary outcome was whether the patient reached maximum disease severity within 28 days. Age and BMI were included as candidate clinical covariates. Patients are the clusters, and the aim is to identify genes whose expression is associated with severity while accounting for persistent between-patient heterogeneity.

\subsection{Variable selection results}

Starting from 5{,}000 candidate genes, our preliminary screen retained 3{,}481 genes (two-sample $t$-test $p < 0.05$), of which 2{,}397 exhibited significant within-cluster correlation. The marginal model LASSO selected 75 genes, while GSHEL selected 37 genes, and the iterative GSHEL selected 35 of the same 37 genes. Neither age nor BMI was selected. GSHEL therefore produced a substantially more parsimonious model that retains the main biological signals described below. The debiased inference identified 10 of the 37 genes as individually significant ($p < 0.05$). The selected genes are summarized in Table S3.

\subsection{Biological interpretation}
We assessed biological relevance by comparing the 37 GSHEL-selected genes with the findings of the original study. Three selected genes, \textit{S100A12}, \textit{CD177}, and \textit{MCEMP1}, are markers of the granulocytic myeloid-derived suppressor cell subtype identified in the original analysis and associated with COVID-19 severity \cite{lasalle2022longitudinal}. GSHEL also selected \textit{SERPINB2}, which was highlighted in the original study because its longitudinal trajectory differed between severe and non-severe patients. These overlaps are consistent with biological patterns reported in the original study.

Several additional selected genes belong to pathways implicated in the original study, including MHC class II activity, interferon response, metabolic and oxidative-stress processes, and myeloid differentiation. Two immunoglobulin genes were also selected. Because immunoglobulin expression in neutrophil-enriched samples may partly reflect differences in cell composition \cite{lasalle2022longitudinal}, we conducted a sensitivity analysis including the immunoglobulin score from the original study as an additional covariate. The sensitivity analysis retained the principal severity-associated markers and both immunoglobulin genes. Both immunoglobulin genes remained selected after this adjustment.


\section{Discussion} \label{sec:disc}
In this article, we have proposed SHEL and GSHE, which use cluster-constant synthetic approximations, constructed from the covariates alone, for cluster effects in high-dimensional clustered
data. The proposed method is computationally feasible and retains the interpretation of the scientifically relevant regression coefficients. The results are supported by oracle inequalities, by consistency for the structural coefficient when the residual heterogeneity is weakly aligned with
the augmented design and for the marginal working parameter otherwise, and by a debiased
inference procedure based on the inverse-Hessian principle.

{\it Penalty functions.} In this article, we focus on the $\ell_1$-penalty function (LASSO). Our method can be adapted to other penalty functions, such as the elastic net or SCAD, and the results of Section \ref{sec:thm} carry over with usual modifications. Although debiased inference is most commonly developed for the LASSO, the underlying bias-correction idea is more general and can be extended to other penalized estimators, and additional technical arguments are typically needed to account for their different shrinkage and optimization properties. For an example of the elastic net, the KKT conditions contain an additional ridge term, so that both the bias-correction step and the inverse-Hessian approximation must be constructed for the ridge-adjusted score map, and additional rate conditions are needed to ensure that the resulting bias is asymptotically negligible for inference \citep{zhang2022elastic}. 

{\it Post-selection inference for ISHEL}. For the SHEL estimator, once $\bfB$ is constructed, the design matrix entering the penalized optimization is fixed. This makes the post-selection problem comparatively tractable, since the KKT conditions lead to a well-defined geometric characterization of the selection event and the associated bias-correction arguments can be developed under a stable design. In contrast, the ISHEL algorithm repeatedly updates the synthetic approximation in a data-dependent manner, so the effective design matrix evolves with the response across iterations. Consequently, the final selection event is determined by a sequence of adaptive design updates rather than by a single fixed optimization problem, and standard post-selection inference tools no longer apply directly. This iterative dependence creates substantial technical challenges for both selective-inference and debiasing-based procedures, and developing valid inference methods for such design-updating algorithms remains an important direction for future work. 


{\it Multiple random effects.} 
We consider only heterogeneous intercepts in this work. Other forms of cluster-level heterogeneity, such as random slopes or nested random effects, can be accommodated in the same framework by augmenting the design with cluster-constant summaries in the $\bfB$ matrix. Extending the theory and the inference procedure to several such components
simultaneously, and selecting among them, are directions for future work.

\clearpage
\newpage

\begin{figure}[h!]
    \centering
    \includegraphics[width=1\linewidth]{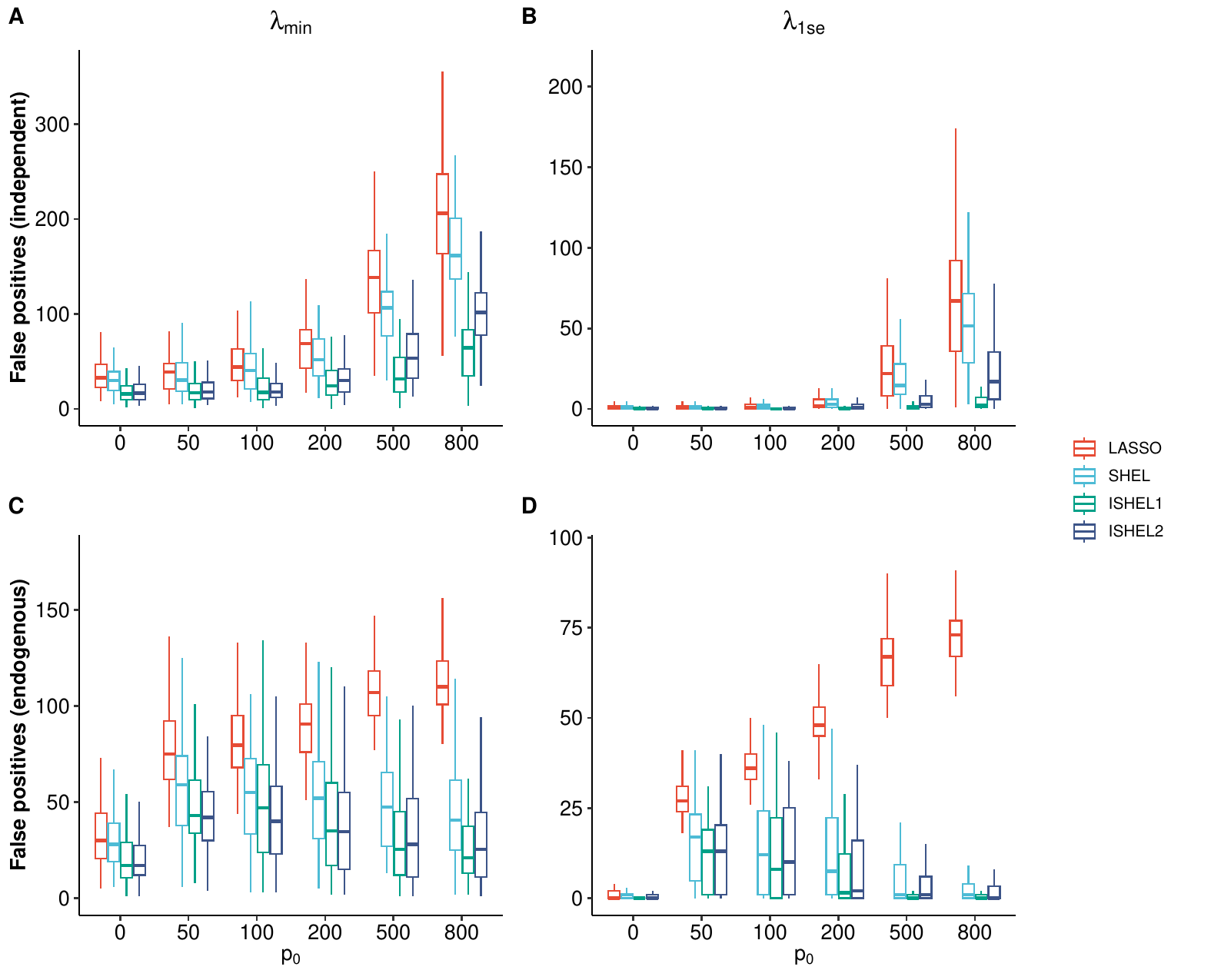}
    \caption{Number of false selections for high-dimensional linear random-effects models with Gaussian random intercepts. }
    \label{Fig_FP}
\end{figure}

\begin{figure}[h!]
    \centering
    \includegraphics[width=1\linewidth]{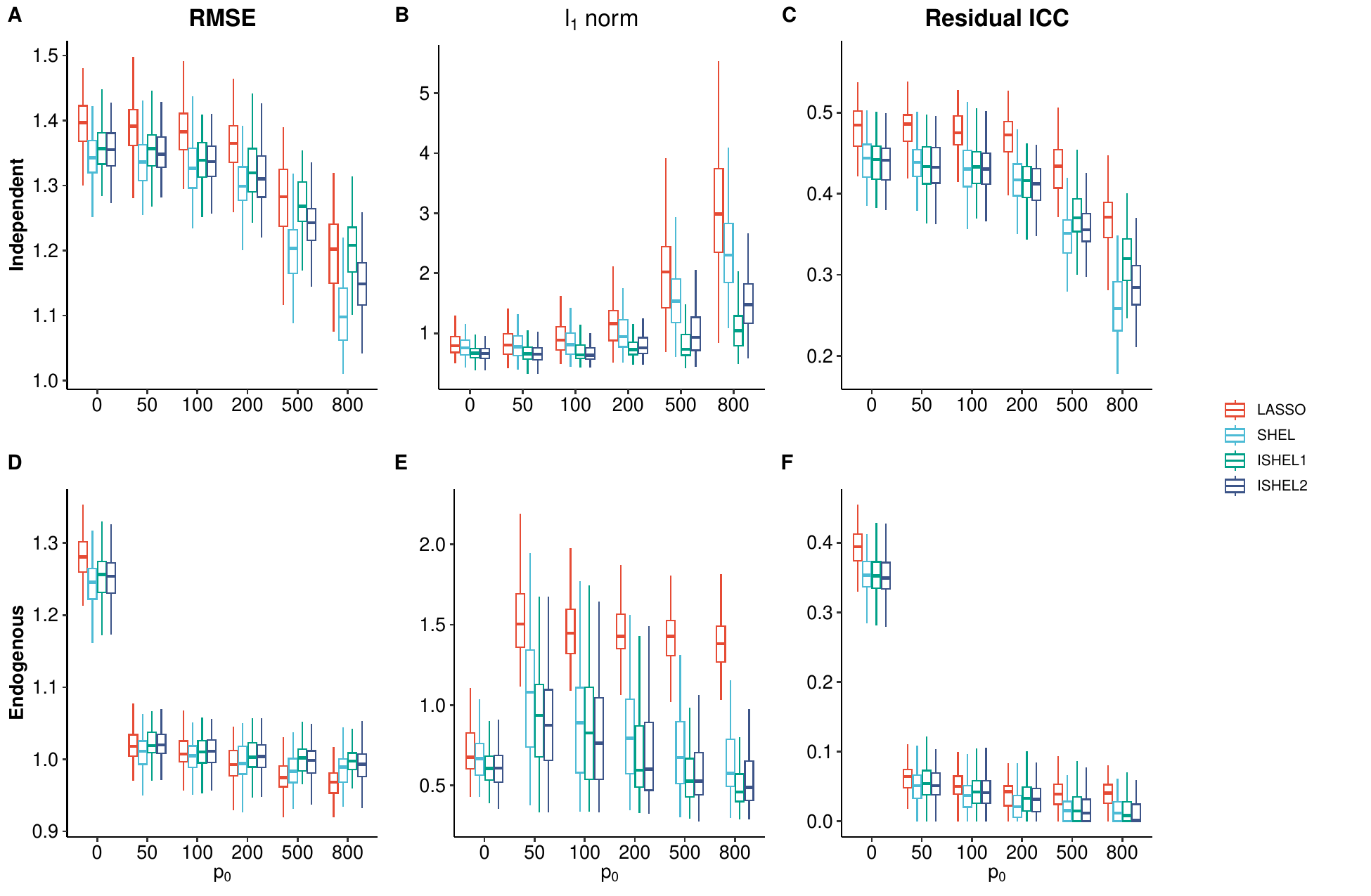}
    \caption{Empirical RMSE, $\lVert \hat{\bfbeta} - \bfbeta^0 \rVert_1$, and residual ICC in high-dimensional linear random-intercepts models with Gaussian random intercepts. }
    \label{Fig_e}
\end{figure}

\begin{table}[!h]
\centering
\caption{\label{tab:debias} Debiased inference under Gaussian random effects. }
\centering
\resizebox{\ifdim\width>\linewidth\linewidth\else\width\fi}{!}{
\fontsize{8}{10}\selectfont
\begin{threeparttable}
\begin{tabular}[t]{llccccccc}
\toprule
\multicolumn{3}{c}{ } & \multicolumn{2}{c}{Pooled type I error} & \multicolumn{2}{c}{Power} & \multicolumn{2}{c}{Median CI length} \\
\cmidrule(l{3pt}r{3pt}){4-5} \cmidrule(l{3pt}r{3pt}){6-7} \cmidrule(l{3pt}r{3pt}){8-9}
Model & Setting & $p_0$ & Debiased & GLMM & Debiased & GLMM & Debiased & GLMM\\
\midrule
 &  & 0 & 0.076 & 0.072 & 1.000 & 1.000 & 0.164 & 0.121\\

 &  & 50 & 0.055 & 0.196 & 1.000 & 1.000 & 0.130 & 0.121\\

 &  & 100 & 0.049 & 0.309 & 1.000 & 1.000 & 0.128 & 0.121\\

 &  & 200 & 0.044 & 0.489 & 1.000 & 1.000 & 0.128 & 0.120\\

 &  & 500 & 0.051 & 0.589 & 1.000 & 1.000 & 0.130 & 0.115\\

 & \multirow{-6}{*}{\raggedright\arraybackslash Endogenous} & 800 & 0.051 & 0.556 & 1.000 & 1.000 & 0.130 & 0.113\\

 &  & 0 & 0.042 & 0.028 & 1.000 & 1.000 & 0.177 & 0.122\\

 &  & 50 & 0.093 & 0.056 & 1.000 & 1.000 & 0.176 & 0.122\\

 &  & 100 & 0.095 & 0.061 & 1.000 & 1.000 & 0.175 & 0.122\\

 &  & 200 & 0.025 & 0.062 & 1.000 & 1.000 & 0.169 & 0.121\\

 &  & 500 & 0.067 & 0.027 & 1.000 & 1.000 & 0.161 & 0.118\\

\multirow{-12}{*}[0.5\dimexpr\aboverulesep+\belowrulesep+\cmidrulewidth]{\raggedright\arraybackslash Linear} & \multirow{-6}{*}{\raggedright\arraybackslash Independent} & 800 & 0.073 & 0.073 & 1.000 & 1.000 & 0.152 & 0.114\\
\cmidrule{1-9}
 &  & 0 & 0.064 & 0.091 & 0.953 & 0.985 & 0.233 & 0.278\\

 &  & 50 & 0.054 & 0.136 & 0.966 & 0.992 & 0.235 & 0.273\\

 &  & 100 & 0.049 & 0.160 & 0.975 & 0.990 & 0.235 & 0.276\\

 &  & 200 & 0.060 & 0.242 & 0.974 & 0.987 & 0.239 & 0.271\\

 &  & 500 & 0.052 & 0.201 & 0.951 & 0.988 & 0.251 & 0.278\\

 & \multirow{-6}{*}{\raggedright\arraybackslash Endogenous} & 800 & 0.050 & 0.197 & 0.942 & 0.992 & 0.258 & 0.285\\

 &  & 0 & 0.049 & 0.075 & 0.955 & 0.979 & 0.232 & 0.287\\

 &  & 50 & 0.053 & 0.060 & 0.955 & 0.985 & 0.231 & 0.287\\

 &  & 100 & 0.053 & 0.064 & 0.956 & 0.985 & 0.231 & 0.288\\

 &  & 200 & 0.057 & 0.061 & 0.963 & 0.988 & 0.229 & 0.285\\

 &  & 500 & 0.046 & 0.054 & 0.960 & 0.992 & 0.227 & 0.281\\

\multirow{-12}{*}[0.5\dimexpr\aboverulesep+\belowrulesep+\cmidrulewidth]{\raggedright\arraybackslash Logistic} & \multirow{-6}{*}{\raggedright\arraybackslash Independent} & 800 & 0.053 & 0.077 & 0.953 & 0.991 & 0.229 & 0.260\\
\bottomrule
\end{tabular}
\begin{tablenotes}
\item \textit{Note: } 
\item Debiased denotes debiased GSHEL introduced in Section \ref{sec:ps}. GLMM denotes the unpenalized low-dimensional mixed-effects refit. Type I error and power are pooled over the selected null and active coordinates, respectively. CI length is summarized by its median over simulation replications.
\end{tablenotes}
\end{threeparttable}}
\end{table}

\clearpage
\newpage

\appendix
\renewcommand{\theequation}{S\arabic{equation}}
\setcounter{equation}{0}

\renewcommand{\thefigure}{S\arabic{figure}}
\setcounter{figure}{0}

\renewcommand{\thetable}{S\arabic{table}}
\setcounter{table}{0}

\begin{center}
{\bf\Large Supplementary Material}    
\end{center}

\section{Proofs}
\subsection{Proof of Theorem 1}
Let $\bfDelta = \hat{\bfbeta} - \bfbeta^\star$. Since $(\bfD \bfdelta^\star)^\top \bfq^\star = 0$, we have $N^{-1/2} \|\bfr^\star\|_2 = R_N^\star$. Optimality of the marginal-model LASSO gives
\be \label{pf1:basic}
\frac{1}{2N}\|\bfX \bfDelta\|_2^2 \le \frac{1}{N}(\bfr^\star + \bfepsilon)^\top \bfX \bfDelta + \lambda \{ \| \bfbeta^\star \|_1 - \| \bfbeta^\star + \bfDelta\|_1\}.
\ee
On the event $\|N^{-1} \bfX^\top \bfepsilon\|_\infty \le \lambda/4$,
\ba
\frac{1}{N} \bfepsilon^\top \bfX \bfDelta \le \frac{\lambda}{4} \|\bfDelta\|_1.
\ea
Moreover,
\ba
\|\bfbeta^\star\|_1 - \|\bfbeta^\star + \bfDelta\|_1 \le \|\bfDelta_{\cfM^\star}\|_1 - \|\bfDelta_{\cfM^{\star,c}}\|_1.
\ea
Hence, 
\be \label{pf:marginal_master}
\frac{1}{2N} \|\bfX \bfDelta\|_2^2 + \frac{3 \lambda}{4} \| \bfDelta_{\cfM^{\star,c}} \|_1 \le \frac{1}{N} (\bfr^\star)^\top \bfX \bfDelta + \frac{5 \lambda}{4} \| \bfDelta_{\cfM^\star} \|_1.
\ee

For part I, Cauchy-Schwarz inequality gives
\be \label{pf1:ineq}
\frac{1}{N} \left| (\bfr^\star)^\top \bfX \bfDelta \right| \le R_N^\star
\frac{\| \bfX \bfDelta \|_2}{\sqrt N}.
\ee
If
\ba
\frac{5\lambda}{4} \|\bfDelta_{\cfM^\star}\|_1 \le 2 R_N^\star \frac{\|\bfX \bfDelta\|_2}{\sqrt N},
\ea
then (\ref{pf:marginal_master}) immediately yields
\ba
\frac{1}{N} \|\bfX \bfDelta\|_2^2 \le C(R_N^\star)^2.
\ea
Otherwise, the residual term can be absorbed into the active-set term,
and (\ref{pf:marginal_master}) implies
\ba
\|\bfDelta_{\cfM^{\star,c}}\|_1 \le 3 \|\bfDelta_{\cfM^{\star}}\|_1.
\ea
The RE condition and Cauchy--Schwarz then give
\ba
\|\bfDelta_{\cfM^{\star}}\|_1 \le \sqrt{M^\star} \|\bfDelta\|_2 \le \sqrt{\frac{M^\star}{\kappa}} \frac{\|\bfX \bfDelta\|_2}{\sqrt N}.
\ea
Substitution into (\ref{pf:marginal_master}) yields
\ba
\frac{1}{N} \|\bfX \bfDelta\|_2^2 \le C \frac{\lambda^2 M^\star}{\kappa}.
\ea
Combining the two cases proves part (I).

For part (II), the additional condition gives
\ba
\frac{1}{N} \left| (\bfr^\star)^\top \bfX \bfDelta \right| \le
\left\| N^{-1} \bfX^\top \bfr^\star \right\|_\infty \|\bfDelta\|_1 \le \frac{\lambda}{4} \|\bfDelta\|_1.
\ea
Together with the stochastic-score bound,
(\ref{pf1:basic}) therefore implies
\ba
\frac{1}{2N} \|\bfX \bfDelta\|_2^2 + \frac{\lambda}{2} \| \bfDelta_{\cfM^{\star,c}} \|_1 \le \frac{3 \lambda}{2} \|\bfDelta_{\cfM^{\star}} \|_1.
\ea
Consequently,
\ba
\|\bfDelta_{\cfM^{\star,c}}\|_1 \le 3\|\bfDelta_{\cfM^{\star}}\|_1.
\ea
Applying the RE condition again yields
\ba
\frac{1}{2N} \|\bfX \bfDelta\|_2^2 \le \frac{3 \lambda}{2} \sqrt{M^\star} \|\bfDelta\|_2 \le \frac{3 \lambda}{2} \sqrt{\frac{M^\star}{\kappa}} \frac{\|\bfX \bfDelta\|_2}{\sqrt N},
\ea
which gives
\ba
\frac{1}{N} \|\bfX \bfDelta\|_2^2 \le C \frac{\lambda^2 M^\star}{\kappa}.
\ea
Finally,
\ba
\|\bfDelta\|_1 \le 4\|\bfDelta_{\cfM^\star}\|_1 \le 4 \sqrt{\frac{M^\star}{\kappa}} \frac{\|\bfX \bfDelta\|_2}{\sqrt N} \le C \frac{\lambda M^\star}{\kappa},
\ea
which proves part (II).

\subsection{Proof of Theorem 2}
Let
\ba
\bfDelta_1 = \hat{\bfbeta} - \bfbeta^0,
\quad
\bfDelta_2 = \hat{\bfgamma} - \bfgamma^0,
\quad
\bfDelta = (\bfDelta_1^\top, \bfDelta_2^\top)^\top,
\ea
and let $c_0 = c_s + c_b < 1/2$. By optimality of $(\hat{\bfbeta}, \hat{\bfgamma})$,
\be \label{pf:thm1_basic}
\begin{split}
\cfD_N( \bftheta^0 + \bfDelta, \bftheta^0 ) \le & - \nabla \cfL_N(\bftheta^0)^\top \bfDelta + \lambda_1 \left( \|\bfbeta^0\|_1 - \|\bfbeta^0 + \bfDelta_1\|_1 \right) \\
&+ \lambda_2 \left( \|\bfgamma^0\|_1 - \|\bfgamma^0 + \bfDelta_2\|_1 \right).
\end{split}
\ee
By (C4), with probability tending to one,
\ba
-\left\{ \nabla \cfL_N(\bftheta^0) - \bfg_N^0 \right\}^\top \bfDelta \le c_s \left( \lambda_1 \|\bfDelta_1\|_1 + \lambda_2 \|\bfDelta_2\|_1 \right).
\ea

Let
$\eta_{ij}^0 = \bfW_{ij}^\top \bftheta^0$. By the mean-value theorem,
\ba
b'(\eta_{ij}^0) - b'(\eta_{ij}^0 + \delta_i) = -v_{ij} \delta_i,
\ea
where $v_{ij} = b''(\eta_{ij}^0 + t_{ij} \delta_i)$ for some
$t_{ij} \in (0, 1)$. Local boundedness of $b''$ therefore gives
\ba
\left| (\bfg_N^0)^\top \bfDelta \right| &\le& C \left\{ \frac{1}{N} \sum_{i=1}^m \sum_{j=1}^{n_i} \delta_i^2 \right\}^{1/2} \frac{\|\bfW \bfDelta\|_2}{\sqrt N} \\
&=& C \Delta_N \frac{\| \bfW \bfDelta \|_2}{\sqrt N}.
\ea
The local lower curvature bound implies
\ba
\cfD_N( \bftheta^0 + \bfDelta, \bftheta^0 ) \ge c \frac{\| \bfW \bfDelta \|_2^2}{N},
\ea
and hence Young's inequality gives
\be \label{pf:thm1:approx}
\left|
(\bfg_N^0)^\top \bfDelta \right| \le \frac{1}{4} \cfD_N( \bftheta^0 + \bfDelta, \bftheta^0 ) + C \Delta_N^2.
\ee

Combining (\ref{pf:thm1_basic})-(\ref{pf:thm1:approx}) with decomposability of the two $\ell_1$ penalties yields
\be \label{pf:thm1:master}
\begin{split}
&\quad \frac{3}{4} \cfD_N( \bftheta^0 + \bfDelta, \bftheta^0 ) + (1 - c_s) \left\{ \lambda_1 \|\bfDelta_{1,\cfM_1^c}\|_1 + \lambda_2 \|\bfDelta_{2,\cfM_2^c}\|_1 \right\} \\
&\le C \Delta_N^2 + (1+c_s) \left\{ \lambda_1 \|\bfDelta_{1,\cfM_1}\|_1 + \lambda_2 \|\bfDelta_{2,\cfM_2}\|_1 \right\}.
\end{split}
\ee

Consider first the case in which
\ba
(1 + c_s) \left\{ \lambda_1 \|\bfDelta_{1,\cfM_1}\|_1 + \lambda_2 \|\bfDelta_{2,\cfM_2}\|_1 \right\} \le 2 C \Delta_N^2.
\ea
Equation (\ref{pf:thm1:master}) immediately gives
\ba
\cfD_N( \hat{\bftheta}, \bftheta^0 ) = O_p(\Delta_N^2)
\ea
and
\ba
\lambda_1 \|\bfDelta_1\|_1 + \lambda_2\|\bfDelta_2\|_1 = O_p(\Delta_N^2).
\ea

Otherwise, the $\Delta_N^2$ term in (\ref{pf:thm1:master}) can be absorbed into the active-set penalty, so $\bfDelta \in \cfC_a(\cfM_1,\cfM_2)$ for a sufficiently large fixed $a$. By (C3) and the integral form of the Taylor expansion,
\ba
\cfD_N( \bftheta^0 + \bfDelta, \bftheta^0 ) \ge \frac{\kappa}{2} \|\bfDelta\|_2^2.
\ea
Moreover,
\ba
\lambda_1 \|\bfDelta_{1,\cfM_1}\|_1 + \lambda_2 \|\bfDelta_{2,\cfM_2}\|_1 \le \left( M_1 \lambda_1^2 + M_2 \lambda_2^2 \right)^{1/2} \|\bfDelta\|_2.
\ea
Substitution into (\ref{pf:thm1:master}) therefore gives
\ba
\cfD_N( \hat{\bftheta}, \bftheta^0 ) = O_p \left( M_1 \lambda_1^2 + M_2 \lambda_2^2 \right)
\ea
and
\ba
\lambda_1 \|\bfDelta_1\|_1 + \lambda_2 \|\bfDelta_2\|_1 = O_p \left( M_1 \lambda_1^2 + M_2 \lambda_2^2 \right).
\ea
Combining the two cases proves (I) and (II). The $\ell_1$ bound for $\bfbeta$ follows immediately by dividing (II) by $\lambda_1$.

\subsection{Proof of Proposition 1}
For the Gaussian identity-link model, $b(\eta) = \eta^2/2$ and $b''(\eta) = 1$. Consequently, 
\ba 
\cfD_N( \bftheta^0 + \bfDelta, \bftheta^0 ) = \frac{1}{2N} \| \bfW \bfDelta \|_2^2. 
\ea
Moreover, 
\ba 
\bfg_N^0 = -\frac{1}{N} \bfW^\top \bfD \bfdelta, 
\ea 
so the approximation term in the proof of Theorem 2 becomes directly 
\ba 
\left| (\bfg_N^0)^\top \bfDelta \right| \le \Delta_N \frac{\| \bfW \bfDelta \|_2}{\sqrt N}. 
\ea 
Thus Proposition 1 follows immediately from the argument of Theorem 2, with the excess loss equal to the quadratic prediction loss.

\subsection{Proof of Corollary 2}
Similar to the proof of Theorem 2
\ba
\bfDelta_1 = \hat{\bfbeta} - \bfbeta^0,
\quad
\bfDelta_2 = \hat{\bfgamma} - \bfgamma^0,
\quad
\bfDelta = (\bfDelta_1^\top, \bfDelta_2^\top)^\top,
\ea
and let $c_0 = c_s + c_b < 1/2$. By (C4)--C5), with probability tending to one,
\ba
\| \nabla_{\bfbeta} \cfL_N(\bftheta^0) \|_\infty \le c_0 \lambda_1, \quad \| \nabla_{\bfgamma} \cfL_N(\bftheta^0) \|_\infty \le c_0\lambda_2.
\ea
Hence
\ba
-\nabla \cfL_N(\bftheta^0)^\top \bfDelta \le c_0 \left( \lambda_1 \|\bfDelta_1\|_1 + \lambda_2 \|\bfDelta_2\|_1 \right).
\ea
Using decomposability of the $\ell_1$ penalty in
(\ref{pf:thm1_basic}) gives
\be \label{pf:cor2_cone}
\begin{split}
&\quad \cfD_N( \bftheta^0 + \bfDelta, \bftheta^0 ) + (1 - c_0) \left\{ \lambda_1 \| \bfDelta_{1,\cfM_1^c} \|_1 + \lambda_2 \| \bfDelta_{2,\cfM_2^c} \|_1 \right\} \\ 
&\le (1 + c_0) \left\{ \lambda_1 \| \bfDelta_{1,\cfM_1} \|_1 + \lambda_2 \| \bfDelta_{2,\cfM_2} \|_1 \right\}.
\end{split}
\ee
Since the excess loss is nonnegative, $\bfDelta \in \cfC_a(\cfM_1, \cfM_2)$ for any fixed $a \ge \frac{1 + c_0}{1 - c_0}$.

By the Cauchy-Schwarz inequality,
\be \label{pf:cor2_active}
&~& \lambda_1 \| \bfDelta_{1,\cfM_1} \|_1 + \lambda_2 \| \bfDelta_{2,\cfM_2} \|_1 \nonumber \\
&\le& \lambda_1 \sqrt{M_1} \| \bfDelta_{1,\cfM_1} \|_2 + \lambda_2 \sqrt{M_2} \| \bfDelta_{2,\cfM_2} \|_2 \nonumber \\
&\le& \left( M_1 \lambda_1^2 + M_2 \lambda_2^2 \right)^{1/2} \|\bfDelta\|_2.
\ee
For $0\le t\le 1$, the second-order Taylor expansion gives
\ba
\cfD_N( \bftheta^0 + \bfDelta, \bftheta^0 ) =
\int_0^1 (1-t) \bfDelta^\top \nabla^2 \cfL_N ( \bftheta^0 + t \bfDelta ) \bfDelta \, dt.
\ea
On the event that $\|\bfDelta\|_2\le r$, condition (C3) therefore yield 
\be \label{pf:cor2_curvature}
\cfD_N( \bftheta^0 + \bfDelta, \bftheta^0) \ge \frac{\kappa}{2} \|\bfDelta\|_2^2.
\ee
Combining (\ref{pf:cor2_cone})-(\ref{pf:cor2_curvature}) gives
\ba
\frac{\kappa}{2} \|\bfDelta\|_2^2 \le (1 + c_0) \left( M_1 \lambda_1^2 + M_2 \lambda_2^2 \right)^{1/2} \|\bfDelta\|_2,
\ea
and hence
\be \label{pf:cor2_l2}
\|\bfDelta\|_2 = O_p \left[ \left( M_1 \lambda_1^2 + M_2 \lambda_2^2 \right)^{1/2} \right].
\ee
Since $M_1 \lambda_1^2 + M_2 \lambda_2^2 = o(1)$, the bound in (\ref{pf:cor2_l2}) is $o_p(1)$, so the local restriction $\|\bfDelta\|_2 \le r$ holds with probability tending to one.

Substituting (\ref{pf:cor2_l2}) into (\ref{pf:cor2_cone}) and (\ref{pf:cor2_active}) yields
\ba
\cfD_N( \hat{\bftheta}, \bftheta^0 ) = O_p \left( M_1 \lambda_1^2 + M_2 \lambda_2^2 \right),
\ea
which proves part (I).

Equation (\ref{pf:cor2_cone}) further implies
\ba
&~& \lambda_1 \|\bfDelta_1\|_1 + \lambda_2 \|\bfDelta_2\|_1 \\
&\le& C \left\{ \lambda_1 \| \bfDelta_{1,\cfM_1} \|_1 + \lambda_2 \| \bfDelta_{2,\cfM_2} \|_1 \right\} \\
&=& O_p \left( M_1 \lambda_1^2 + M_2 \lambda_2^2 \right),
\ea
which proves part (II). The bound for $\bfbeta$ follows by dividing the last part by $\lambda_1$. 

\subsection{Proof of Proposition 2}
Similar to the previous proofs, let
\ba
\bfDelta = \hat{\bftheta} - \widetilde{\bftheta}, \quad \bfDelta_1 = \hat{\bftheta} - \widetilde{\bfbeta}, \quad \bfDelta_2 = \hat{\bfgamma} - \widetilde{\bfgamma}.
\ea
By the definition of the marginal-model parameter, $\E \left\{ \nabla \cfL( \widetilde{\bftheta} ) \right\} = \bfzero$. The centered score condition therefore implies, with probability tending to one,
\ba
\| \nabla_{\bfbeta} \cfL( \widetilde{\bftheta} ) \|_\infty \le c_s \lambda_1, \quad
\| \nabla_{\bfgamma} \cfL( \widetilde{\bftheta} ) \|_\infty \le c_s \lambda_2.
\ea
The basic inequality and decomposability give
\ba
&~& \cfD_N( \widetilde{\bftheta} + \bfDelta, \widetilde{\bftheta} ) + (1-c_s) \left\{ \lambda_1 \| \bfDelta_{1,\widetilde{\cfM}_1^c} \|_1 + \lambda_2 \| \bfDelta_{2,\widetilde{\cfM}_2^c} \|_1 \right\}
\\
&\le& (1+c_s) \left\{ \lambda_1 \| \bfDelta_{1,\widetilde{\cfM}_1} \|_1 + \lambda_2 \| \bfDelta_{2,\widetilde{\cfM}_2} \|_1 \right\}.
\ea
Hence, $\bfDelta$ belongs to the corresponding weighted cone. Moreover,
\ba
\lambda_1 \| \bfDelta_{1,\widetilde{\cfM}_1} \|_1 + \lambda_2 \| \bfDelta_{2,\widetilde{\cfM}_2} \|_1 \le \left( \widetilde M_1 \lambda_1^2 + \widetilde M_2 \lambda_2^2 \right)^{1/2} \|\bfDelta\|_2.
\ea
The weighted RE condition and the same localization argument as in the proof of Theorem 2 yield
\ba
\|\bfDelta\|_2 = O_p \left[ \left( \widetilde M_1 \lambda_1^2 + \widetilde M_2 \lambda_2^2 \right)^{1/2} \right].
\ea
Substituting the above bound into the basic completes the proof.

\subsection{Proof of Theorem 3}
By definition of the marginal working parameter, 
\ba 
\E \left\{ \nabla \cfL_N(\widetilde{\bftheta}) \right\} = \bfzero.
\ea 
Hence the assumed score concentration at $\widetilde{\bftheta}$ gives 
\ba 
- \nabla \cfL_N(\widetilde{\bftheta})^\top \widetilde{\bfDelta} \le c_s \left\{ \lambda_1 \|\widetilde{\bfDelta}_1 \|_1 + \lambda_2 \| \widetilde{\bfDelta}_2 \|_1 \right\} 
\ea 
with probability tending to one. The basic inequality is therefore the same as that used in the proof of Corollary 1, with $(\bftheta^0, \cfM_1, \cfM_2, M_1, M_2)$ replaced by $(\widetilde{\bftheta}, \widetilde{\cfM}_1, \widetilde{\cfM}_2, \widetilde M_1, \widetilde M_2)$ and $c_0$ replaced by $c_s$. The corresponding weighted cone condition, weighted RE condition, and Cauchy--Schwarz argument consequently yield 
\ba 
\cfD_N( \hat{\bftheta}, \widetilde{\bftheta} ) = O_p \left( \widetilde M_1 \lambda_1^2 + \widetilde M_2 \lambda_2^2 \right) 
\ea 
and 
\ba 
\lambda_1 \| \hat{\bfbeta} - \widetilde{\bfbeta} \|_1 + \lambda_2 \| \hat{\bfgamma} - \widetilde{\bfgamma} \|_1 = O_p \left( \widetilde M_1 \lambda_1^2 + \widetilde M_2 \lambda_2^2 \right). 
\ea

\subsection{Proof of Theorem 4}
The first-order Taylor expansion of $\nabla \cfL_N(\widetilde{\bftheta})$ at $\hat{\bftheta}$ gives
\be \label{pf6:scoreexp}
\nabla \cfL_N(\hat{\bftheta}) = \nabla \cfL_N(\widetilde{\bftheta}) + \bar{\bfH}_N \bfDelta,
\ee
where $\bar{\bfH}_N$ is defined in (I4). Therefore,
\be \label{pf6:decomp}
\hat \beta_l^{d} - \widetilde\beta_l = -\hat{\bfa}_l^\top \nabla \cfL_N(\widetilde{\bftheta}) + \{ \bfe_l^\top - \hat{\bfa}_l^\top \bar{\bfH}_N \} \bfDelta.
\ee
Based on the symmetry of $\widehat{\bfH}$, we have
\ba
\{ \bfe_l^\top - \hat{\bfa}_l^\top \bar{\bfH}_N \} \bfDelta = \bfr_l^\top \bfDelta + \hat{\bfa}_l^\top (\widehat{\bfH} - \bar{\bfH}_N) \bfDelta.
\ea
By the blockwise remainder conditions in (I4), this is $o_p(m^{-1/2})$. Adding and subtracting $\bfa_l^0$ in the first term of (\ref{pf6:decomp}) gives
\be \label{pf6:main}
\begin{split}
-\hat{\bfa}_l^\top \nabla \cfL_N(\widetilde{\bftheta}) = &-(\bfa_l^0)^\top \nabla \cfL_N(\widetilde{\bftheta}) - (\hat{\bfa}_l - \bfa_l^0)^\top \nabla \cfL_N(\widetilde{\bftheta}).
\end{split}
\ee
The second term in (\ref{pf6:main}) is $o_p(m^{-1/2})$ by (I4). By the definitions of $\phi_{l,ij}$ and $\Phi_{l,i}$,
\ba
-(\bfa_l^0)^\top \nabla \cfL_N(\widetilde{\bftheta}) = \frac{1}{N} \sum_{i=1}^m \sum_{j=1}^{n_i} \phi_{l,ij} = \frac{1}{m} \sum_{i=1}^m \Phi_{l,i}.
\ea
Consequently,
\ba
\sqrt m (\hat \beta_l^{d} - \widetilde \beta_l) = \frac{1}{\sqrt m} \sum_{i=1}^m \Phi_{l,i} + o_p(1).
\ea
Because $\nabla \cfL (\widetilde{\bftheta}) = \bfzero$, $\E(\Phi_{l,i}) = 0$. Conditions (I1) and (I5) therefore imply, by the Lyapunov central limit theorem \citep{billingsley2013convergence},
\ba
\frac{1}{\sqrt m} \sum_{i=1}^m \Phi_{l,i} \xrightarrow{d} N(0, V_l).
\ea

It remains to establish variance consistency. By the law of large numbers and (I1),
\ba
\frac{1}{m}\sum_{i=1}^m \Phi_{l,i}^2 \xrightarrow{P} V_l.
\ea
Condition (I5) and the Cauchy-Schwarz inequality imply
\ba
\frac{1}{m} \sum_{i=1}^m | \hat \Phi_{l,i}^2 - \Phi_{l,i}^2 | = o_p(1).
\ea
Moreover, $\bar \Phi_l=o_p(1)$. Hence
\ba
\hat V_l =
\frac{1}{m} \sum_{i=1}^m ( \hat \Phi_{l,i} - \bar\Phi_l )^2 \xrightarrow{P} V_l.
\ea
By Chebyshev's inequality, part (II) of the theorem follows from Slutsky's theorem.

\clearpage

\section{Algorithm and implementation} \label{sec:alg}
\subsection{Implementation of GSHEL}
The GSHEL estimator (8) can be re-expressed as
\be \label{GSHEL2}
\hat{\bftheta} = \underset{\bftheta \in \mathbb{R}^{p+p_0}}{\arg\min} \left\{ \frac{1}{N} \sum_{i,j} l(\bfW_{ij}^\top \bftheta; y_{ij}) + \lambda_1 \sum_{k=1}^{p+p_0} w_k \lvert \theta_k \rvert \right\},
\ee
where $w_k = 1$ for the $\bfbeta$-coordinates and $w_k = \lambda_2/\lambda_1$ for the $\bfgamma$-coordinates. Thus, GSHEL can be viewed as a weighted LASSO estimator. To reduce the computational burden, we set $\lambda_2 = \lambda_1 \sqrt{\log(p_0)/\log(p)}$, and select $\lambda_1$ by $K$-fold cross-validation. 

To construct the cluster-constant matrix $\bfB$, our default construction screens covariates for between-cluster heterogeneity and uses their empirical cluster means as candidate synthetic controls. For a continuous covariate $X_l$, we perform an ANOVA, and for a categorical covariate $X_l$, we fit a generalized random-intercepts model and test whether the between-cluster variance is zero. The screening threshold is treated as a tuning parameter. Alternative constructions include PCA or sparse PCA applied to cluster-level summaries. See details in Algorithm \ref{algo1}. 

\subsection{Iterative SHEL and GSHEL} \label{subsec:ISHEL}
The consistency of the GSHEL estimator relies on the sparsity condition, i.e.,  $M_1 \lambda_1^2 + M_2 \lambda_2^2 = o(1)$. This condition is often plausible when endogenous covariates are present, as is common in observational studies \citep{grilli2011role}. However, when $\bfB \indep \bfalpha$, the optimal solution $\bfgamma^0$ need not be sparse, because the synthetic approximation may be driven mainly by spurious correlations between the cluster means of the covariates and $\bfalpha$. More precisely, \cite{mulayoff2019minimal} showed that $\delta_m = O_p\left( \left( p_0 / m \right)^{-\frac{a}{1-a}} \right)$, where $a = M_2/m \in (0, 1)$ is the sparsity fraction. Since the primary inferential target is $\bfbeta$, the coefficient vector $\bfgamma$ can be viewed as a nuisance parameter. Motivated by this observation, we propose iterative versions of GSHEL, denoted IGSHEL (or ISHEL for linear models), in which the final fitting step in Algorithm 1 is repeated while conditioning on the estimated synthetic approximation from the previous iteration.

Let $\widehat{\bfD \bfalpha}^{(s)} = \bfB \hat{\bfgamma}^{(s)}$ denote the estimated synthetic approximation at iteration $s$. At each step, we augment the current design with $\widehat{\bfD \bfalpha}^{(s-1)}$ from the previous iteration and refit the penalized model. The algorithm stops when $\left\| \widehat{\bfD\bfalpha}^{(s)} - \widehat{\bfD\bfalpha}^{(s-1)} \right \|_2^2 < e_{\mathrm{thr}}$. This iterative scheme can yield a sparser final solution. However, because the design is updated in a data-dependent manner across iterations, the post-selection inference procedures developed in Section 4 do not directly apply to IGSHEL. See details in Algorithm \ref{algo2}. 

\begin{algorithm}
\caption{Synthetic heterogeneous-effects LASSO (SHEL) and generalized synthetic heterogeneous-effects LASSO (GSHEL).}
\label{algo1}
\begin{algorithmic}[1]
\State Set significance level $\alpha$
\State $j \gets 1$
\For{$l = 1, \cdots, p$}
    \If{$X_l$ is continuous}
        \State Perform an ANOVA
    \ElsIf{$X_l$ is categorical}
        \State Fit a generalized random-intercepts model
        \State Test whether the between-cluster variance is zero
    \EndIf
    \If{p-value $< \alpha$}
        \State Compute the cluster means of $X_l$, denoted by $\bar{\bfX}_l$
        \State $\bfB_{\cdot j} \gets \bar{\bfX}_l$; $j \gets j+1$
    \EndIf
\EndFor
\State $p_0 \gets j-1$
\State Fit (8) with $\lambda_1$ selected by $K$-fold cross-validation
\State Apply the debiased inference procedure 
\end{algorithmic}
\end{algorithm}

\begin{algorithm}
\caption{Iterative synthetic heterogeneous-effects LASSO (ISHEL) and generalized synthetic heterogeneous-effects LASSO (IGSHEL).}
\label{algo2}
\begin{algorithmic}[1]
\State Run Algorithm 1 without the inference step
\State Initialize $\widehat{\bfD\bfalpha}^{\,(0)} \gets \bfB \hat{\bfgamma}^{(0)}$
\State Set convergence threshold $e_{\mathrm{thr}}$; set $e \gets \infty$ and $s \gets 1$
\While{$e \ge e_{\mathrm{thr}}$}
    \State Augment the design matrix with $\widehat{\bfD\bfalpha}^{\,(s-1)}$
    \State Fit (18) or (19), with $\lambda_1$ selected by $K$-fold cross-validation
    \State Update $\widehat{\bfD\bfalpha}^{\,(s)} \gets \bfB \hat{\bfgamma}^{(s)}$
    \State $e \gets \left\|\widehat{\bfD\bfalpha}^{\,(s)} - \widehat{\bfD\bfalpha}^{\,(s-1)}\right\|_2^2$
    \State $s \gets s+1$
\EndWhile
\end{algorithmic}
\end{algorithm}

\clearpage
\newpage

\section{Additional simulation results}
\subsection{Definitions of performance measures}
\begin{table}[h!]
\caption{Measures of performance in simulation studies.}
    \centering
    \begin{tabular}{cp{0.7\linewidth}}
    \toprule
    Metrics & Definition \\
    \midrule
    False positives & the number of selected covariates that corresponded to the zero coefficients \\
    True positives & the number of selected covariates that correspond to the nonzero coefficients \\
    Root mean squared error & square root of the quadratic mean of the differences between $Y_{ij}$ and $\hat{Y}_{ij}$ \\
    Residue ICC & residue intracluster correlation, obtained by fitting a random-intercept model on the fitted residuals \\
    Type I error & the pooled proportion of null coordinates rejected at a nominal significance level of 0.05 \\
    Sensitivity & sensitivity of in-sample prediction \\
    Specificity & specificity of in-sample prediction \\
    Power & the pooled proportion of nonzero coefficients is tested significant at a significance level of 0.05 \\
    Median CI length & median length of the 95\% confidence intervals for variables selected by SHEL or GSHEL \\
    \bottomrule
    \end{tabular}
    \label{Measures}
\end{table}

\clearpage
\newpage

\subsection{Additional simulation results for variable selection, parameter estimation, and prediction}
\begin{figure}[h!]
    \centering
    \includegraphics[width=1\linewidth]{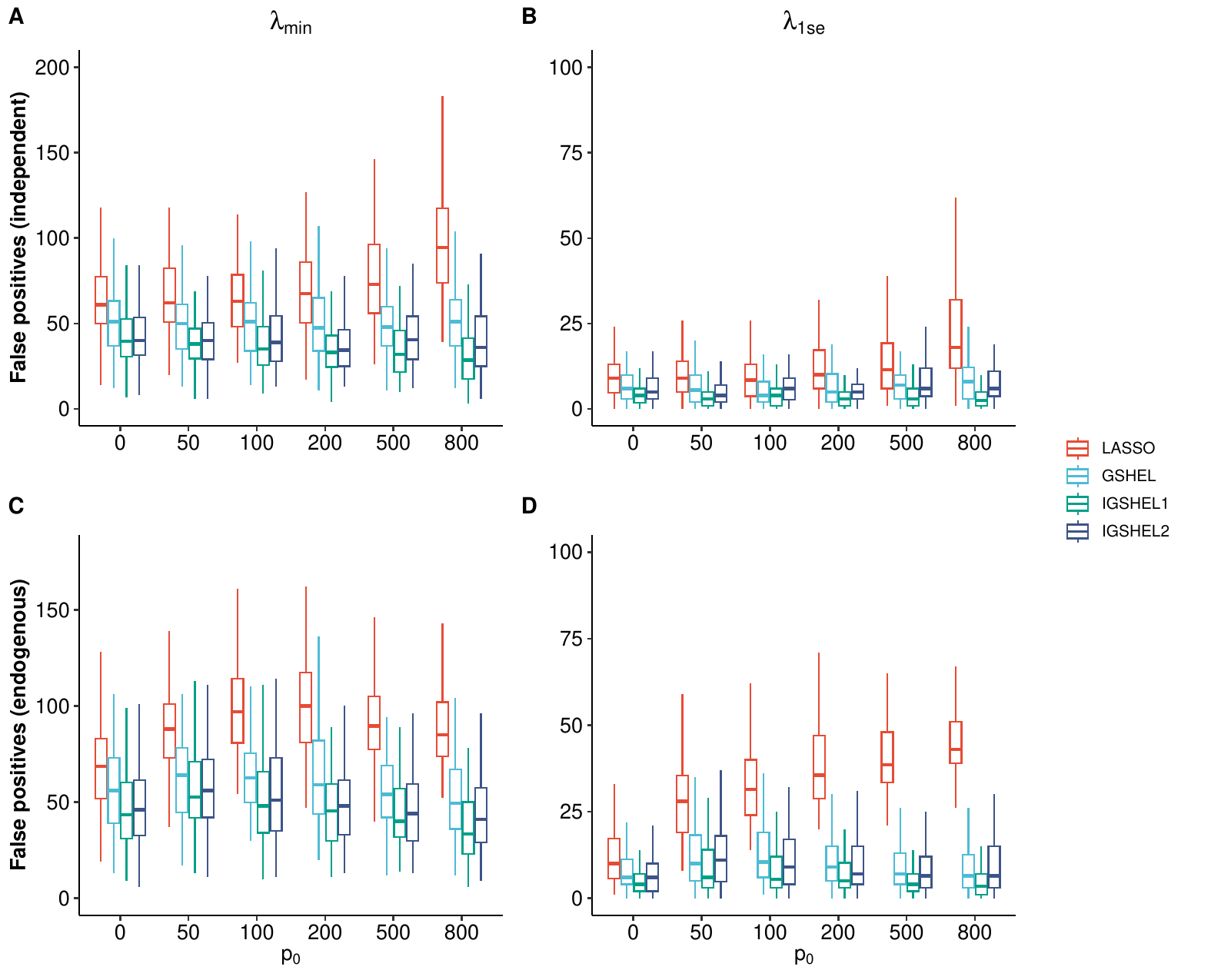}
    \caption{Number of false selections in high-dimensional logistic random-intercepts models with Gaussian random intercepts. }
    \label{Fig_FP2}
\end{figure}

\begin{figure}[h!]
    \centering
    \includegraphics[width=1\linewidth]{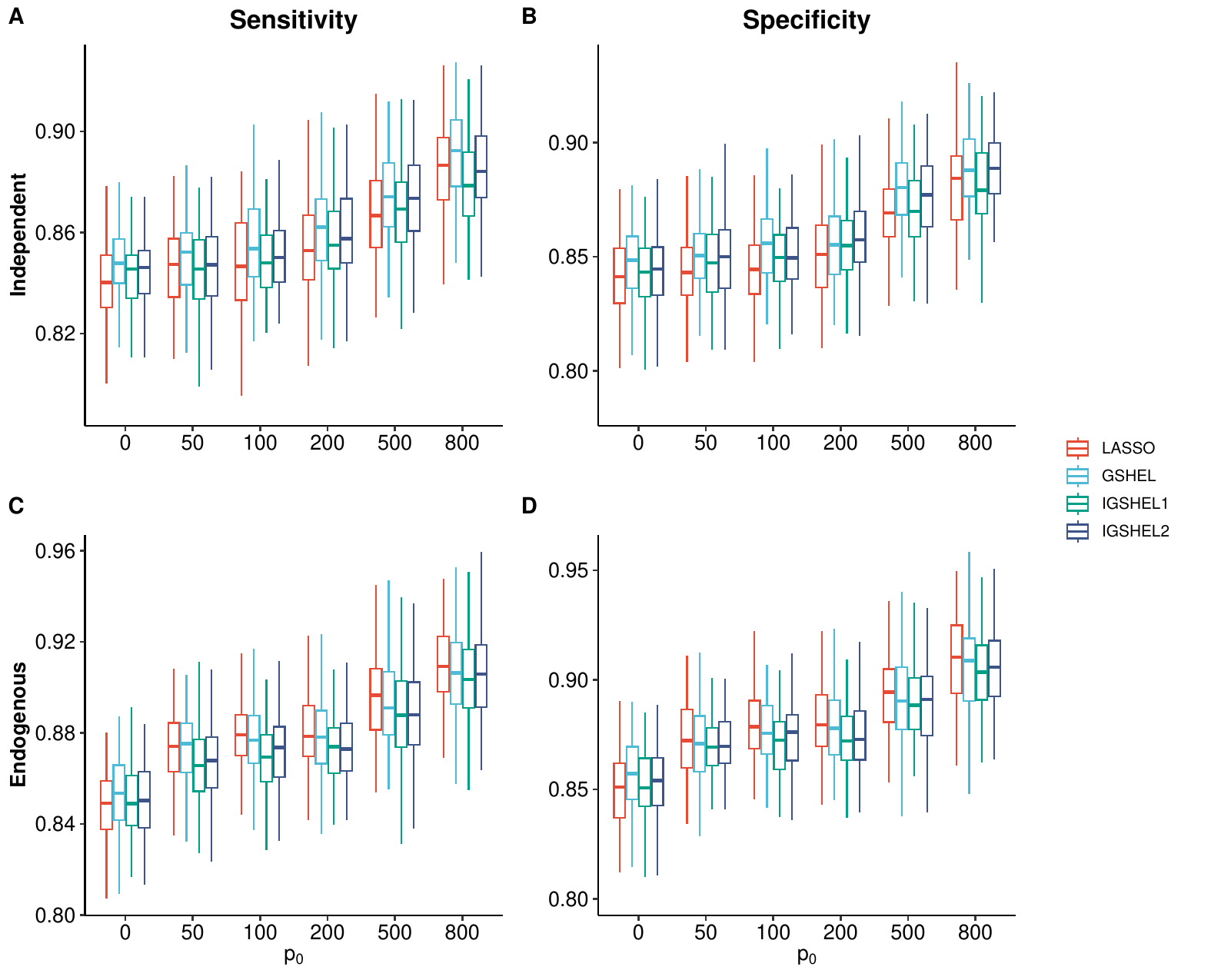}
    \caption{Sensitivity and specificity in high-dimensional logistic random-intercepts models with Gaussian random intercepts. }
    \label{Fig_e2}
\end{figure}

\begin{figure}[h!]
    \centering
    \includegraphics[width=1\linewidth]{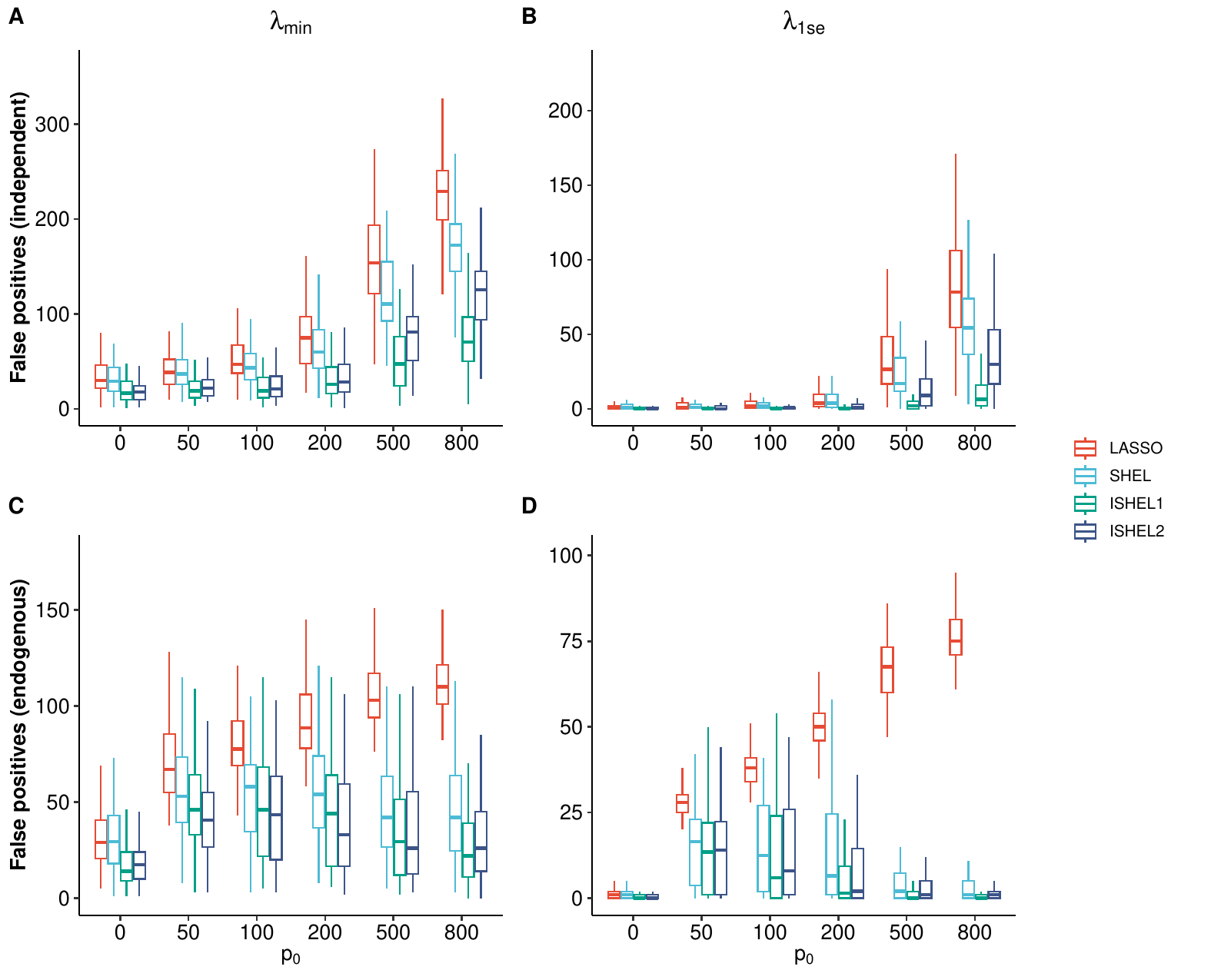}
    \caption{Number of false selections for high-dimensional linear random-effects models with Gaussian mixture random intercepts. }
    \label{Fig_FP3}
\end{figure}

\begin{figure}[h!]
    \centering
    \includegraphics[width=1\linewidth]{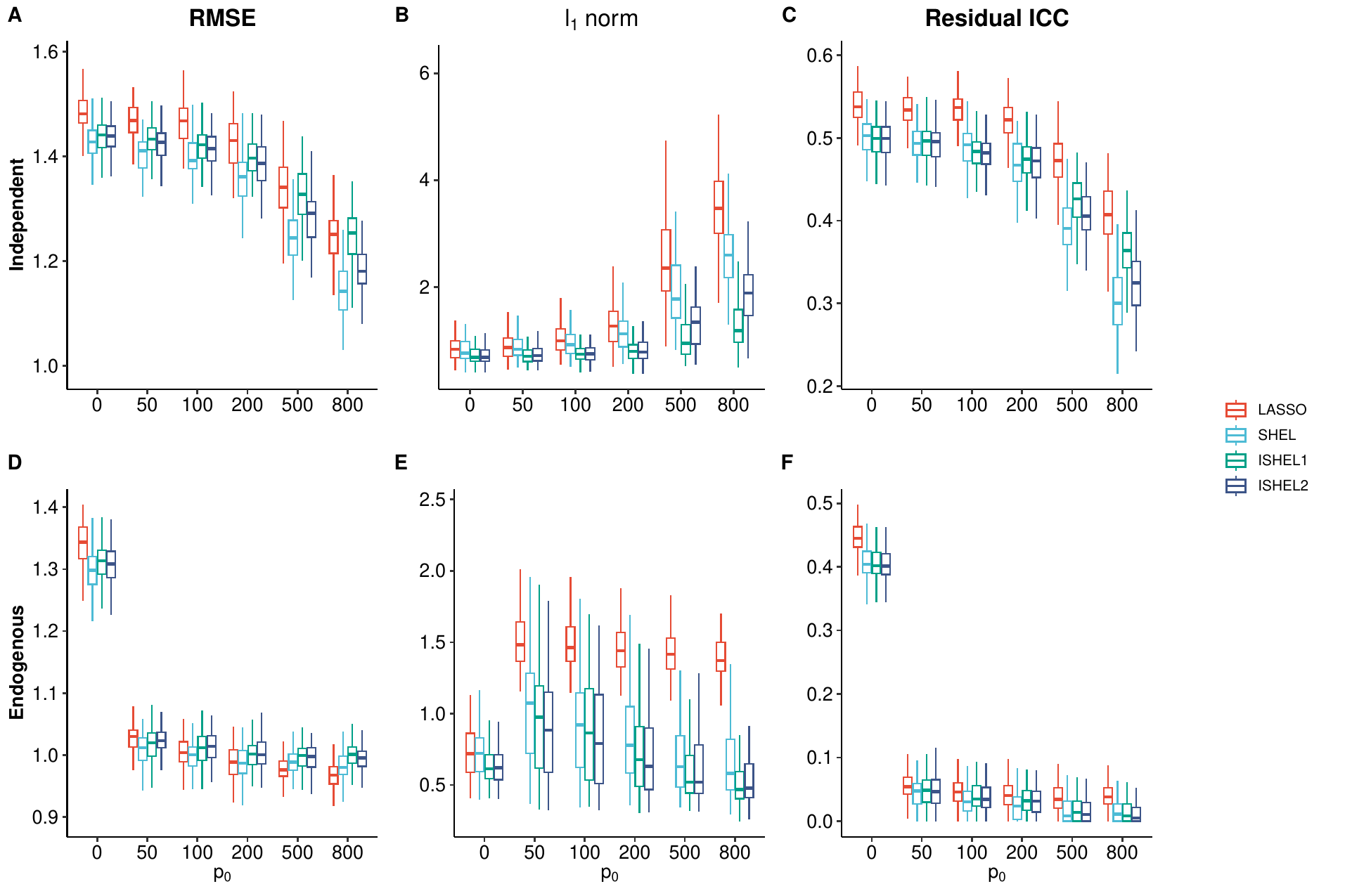}
    \caption{Empirical RMSE, $\lVert \hat{\bfbeta} - \bfbeta^0 \rVert_1$, and residual ICC in high-dimensional linear random-intercepts models with Gaussian mixture random intercepts. }
    \label{Fig_e3}
\end{figure}

\begin{figure}[h!]
    \centering
    \includegraphics[width=1\linewidth]{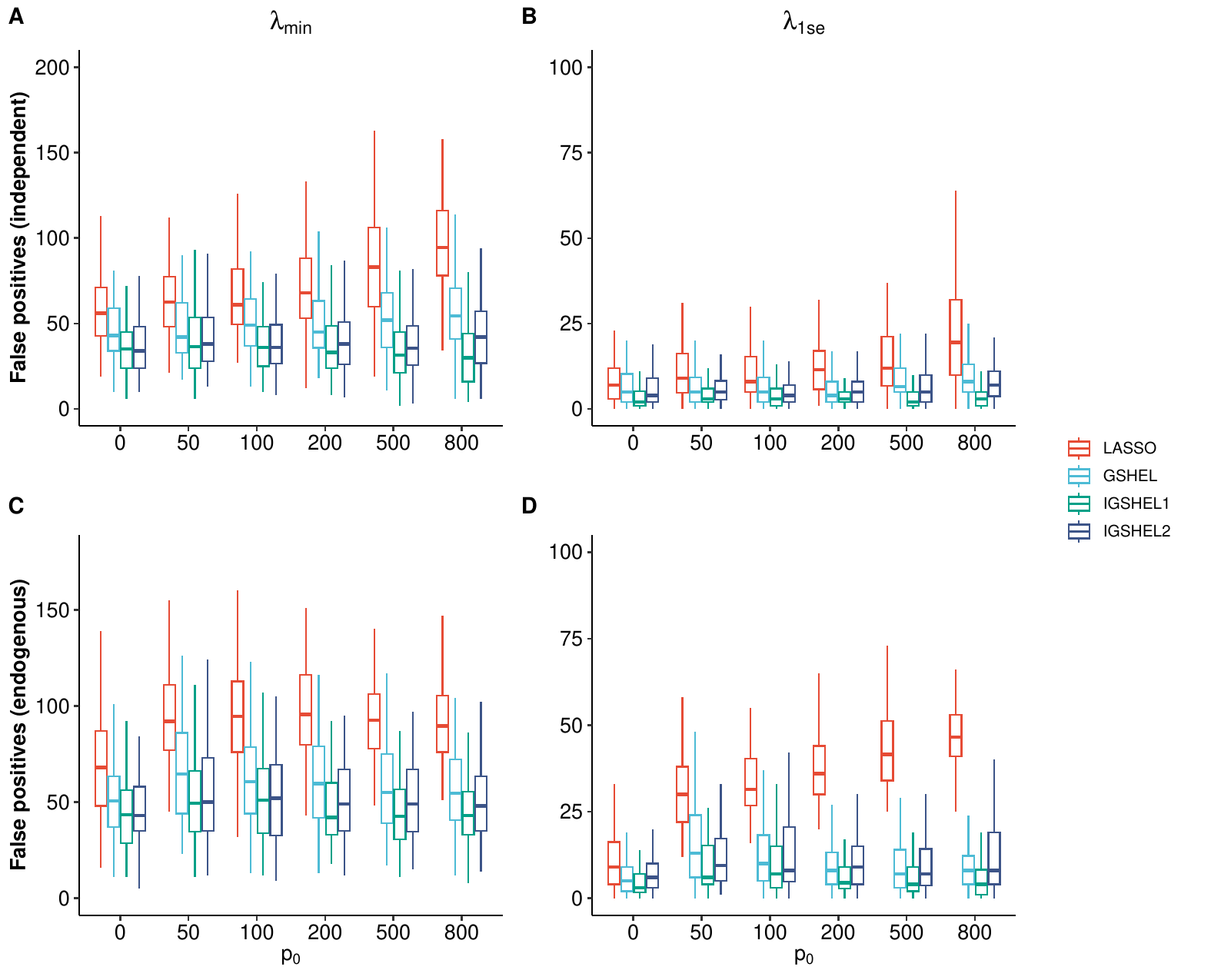}
    \caption{Number of false selections in high-dimensional logistic random-intercepts models with Gaussian mixture random intercepts. }
    \label{Fig_FP4}
\end{figure}

\begin{figure}[h!]
    \centering
    \includegraphics[width=1\linewidth]{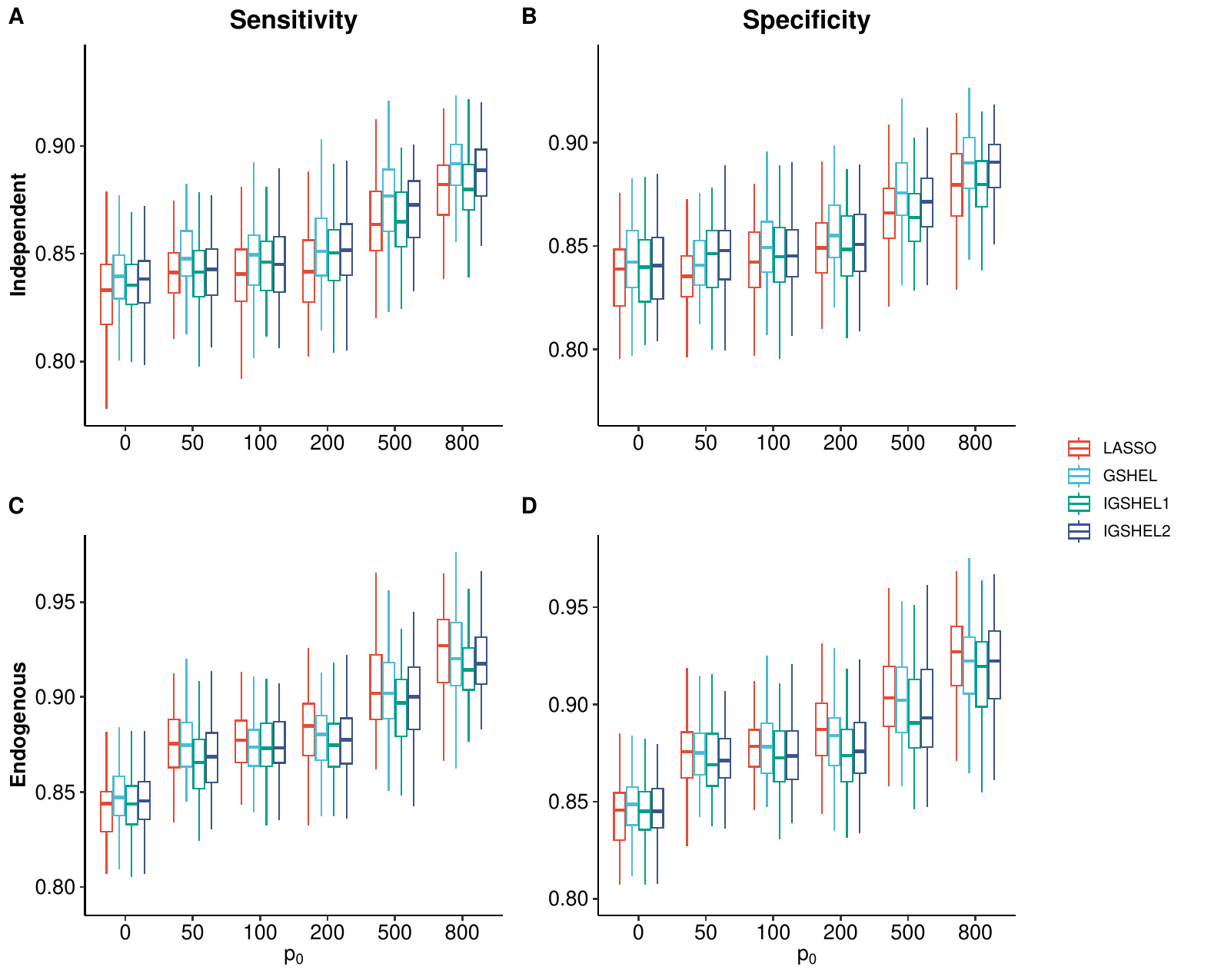}
    \caption{Sensitivity and specificity in high-dimensional logistic random-intercepts models with Gaussian mixture random intercepts. }
    \label{Fig_e4}
\end{figure}

\clearpage
\newpage

\subsection{Additional simulation results for post-selection inference}

\begin{table}[!h]
\centering
\caption{\label{tab:debias_mix} Debiased inference under Gaussian-mixture random effects. }
\centering
\resizebox{\ifdim\width>\linewidth\linewidth\else\width\fi}{!}{
\fontsize{8}{10}\selectfont
\begin{threeparttable}
\begin{tabular}[t]{llccccccc}
\toprule
\multicolumn{3}{c}{ } & \multicolumn{2}{c}{Pooled type I error} & \multicolumn{2}{c}{Power} & \multicolumn{2}{c}{Median CI length} \\
\cmidrule(l{3pt}r{3pt}){4-5} \cmidrule(l{3pt}r{3pt}){6-7} \cmidrule(l{3pt}r{3pt}){8-9}
Model & Setting & $p_0$ & Debiased & GLMM & Debiased & GLMM & Debiased & GLMM\\
\midrule
 &  & 0 & 0.055 & 0.055 & 1.000 & 1.000 & 0.170 & 0.122\\

 &  & 50 & 0.055 & 0.230 & 1.000 & 1.000 & 0.129 & 0.121\\

 &  & 100 & 0.056 & 0.375 & 1.000 & 1.000 & 0.127 & 0.121\\

 &  & 200 & 0.062 & 0.503 & 1.000 & 1.000 & 0.127 & 0.121\\

 &  & 500 & 0.064 & 0.632 & 1.000 & 1.000 & 0.130 & 0.115\\

 & \multirow{-6}{*}{\raggedright\arraybackslash Endogenous} & 800 & 0.054 & 0.577 & 1.000 & 1.000 & 0.130 & 0.112\\

 &  & 0 & 0.053 & 0.064 & 1.000 & 1.000 & 0.185 & 0.123\\

 &  & 50 & 0.068 & 0.034 & 1.000 & 1.000 & 0.186 & 0.123\\

 &  & 100 & 0.035 & 0.026 & 1.000 & 1.000 & 0.184 & 0.123\\

 &  & 200 & 0.048 & 0.048 & 1.000 & 1.000 & 0.179 & 0.122\\

 &  & 500 & 0.034 & 0.034 & 1.000 & 1.000 & 0.165 & 0.119\\

\multirow{-12}{*}[0.5\dimexpr\aboverulesep+\belowrulesep+\cmidrulewidth]{\raggedright\arraybackslash Linear} & \multirow{-6}{*}{\raggedright\arraybackslash Independent} & 800 & 0.048 & 0.048 & 1.000 & 1.000 & 0.156 & 0.116\\
\cmidrule{1-9}
 &  & 0 & 0.054 & 0.079 & 0.968 & 0.988 & 0.232 & 0.286\\

 &  & 50 & 0.060 & 0.150 & 0.968 & 0.988 & 0.235 & 0.280\\

 &  & 100 & 0.055 & 0.206 & 0.969 & 0.979 & 0.237 & 0.279\\

 &  & 200 & 0.054 & 0.263 & 0.959 & 0.980 & 0.242 & 0.280\\

 &  & 500 & 0.062 & 0.263 & 0.953 & 0.994 & 0.257 & 0.286\\

 & \multirow{-6}{*}{\raggedright\arraybackslash Endogenous} & 800 & 0.043 & 0.241 & 0.931 & 0.994 & 0.266 & 0.300\\

 &  & 0 & 0.059 & 0.059 & 0.956 & 0.976 & 0.232 & 0.295\\

 &  & 50 & 0.051 & 0.053 & 0.968 & 0.992 & 0.231 & 0.298\\

 &  & 100 & 0.048 & 0.058 & 0.958 & 0.985 & 0.229 & 0.296\\

 &  & 200 & 0.051 & 0.056 & 0.955 & 0.988 & 0.227 & 0.294\\

 &  & 500 & 0.048 & 0.056 & 0.960 & 0.986 & 0.224 & 0.289\\

\multirow{-12}{*}[0.5\dimexpr\aboverulesep+\belowrulesep+\cmidrulewidth]{\raggedright\arraybackslash Logistic} & \multirow{-6}{*}{\raggedright\arraybackslash Independent} & 800 & 0.049 & 0.045 & 0.953 & 0.988 & 0.227 & 0.272\\
\bottomrule
\end{tabular}
\begin{tablenotes}
\item \textit{Note: } 
\item Debiased denotes debiased GSHEL introduced in Section 4. GLMM denotes the unpenalized low-dimensional mixed-effects refit. Type I error and power are pooled over the selected null and active coordinates, respectively. CI length is summarized by its median over simulation replications.
\end{tablenotes}
\end{threeparttable}}
\end{table}

\clearpage
\newpage

\section{Results for real data analysis}
\begin{table}[h!]
\centering
\caption{Biological categories of the 37 genes selected 
by GSHEL.}
\label{tab:gshel_genes}
\small
\begin{tabular}{llp{5.5cm}}
\toprule
\textbf{Category} & \textbf{Selected Genes} 
& \textbf{Biological Relevance} \\
\midrule
G-MDSC markers (3) 
& \textit{S100A12}$^*$, \textit{CD177}, \textit{MCEMP1} 
& Core NMF5 markers; top features in original 
  LASSO model \\
Longitudinal signal (1) 
& \textit{SERPINB2}$^*$
& Diverging trajectory between severe and 
  non-severe patients \\
Neutrophil chemotaxis (1) 
& \textit{CXCL2} 
& Neutrophil migration pathway enriched 
  in severe disease \\
MHC class II (2) 
& \textit{HLA-DMA}, \textit{HLA-DMB} 
& Enriched in non-severe patients \\
Interferon response (4) 
& \textit{RARRES3}$^*$, \textit{IRF2BP2}, 
  \textit{UPP1}, \textit{CTSW} 
& Interferon pathways enriched in infection 
  and associated with severity \\
Metabolism / ROS (3) 
& \textit{NQO2}, \textit{GPAT4}, \textit{ALDH1A1} 
& Metabolic pathways associated with 
  severe disease \\
Myeloid differentiation (2) 
& \textit{HOTAIRM1}, \textit{PDE4D} 
& Known roles in neutrophil development 
  and activation \\
Humoral immunity (2) 
& \textit{IGHM}, \textit{IGLVI-70}$^*$
& Humoral immune responses associated 
  with severity \\
Other protein-coding (11) 
& \textit{CCR3}, \textit{TSPAN2}$^*$, 
  \textit{TDRD9}$^*$, 
& Immune signaling, regulation, \\
& \textit{USP11}, \textit{METTL7B}, \textit{COA4}, 
& and cell metabolism \\
& \textit{YBEY}, \textit{MLH3}, 
  \textit{RPGRIP1}$^*$, & \\
& \textit{NOV}, \textit{AC008763.3}$^*$ & \\
Non-coding / 
& \textit{LINC00877}, \textit{RPL21P44},  
& Functional interpretation limited \\
unannotated (8)
& \textit{AC105749.1}, \textit{AC037198.2}$^*$, & \\
& \textit{AC108134.3}, \textit{AC023906.5}, & \\
& \textit{AC009303.4}, 
  \textit{ENSG00000288064}$^*$ & \\
\bottomrule
\multicolumn{3}{l}{\footnotesize $^*$Individually 
significant ($p < 0.05$) by the debiased GSHEL.} \\
\multicolumn{3}{l}{\footnotesize Ensembl IDs are 
retained for genes without established gene symbols.} \\
\end{tabular}
\end{table}

\bibliographystyle{apalike}
\bibliography{ref}

\end{document}